%% file: main.tex
\documentclass{article}
\usepackage[T1]{fontenc}
\usepackage[dvipdfmx]{graphicx}
\usepackage{arxiv}
\usepackage[utf8]{inputenc} 
\usepackage[T1]{fontenc}    
\usepackage{hyperref}       
\usepackage{bm}
\usepackage{algorithm,algorithmic}
\usepackage{booktabs}
\usepackage{amsmath,amssymb,amsthm}
\usepackage{url}
\usepackage{doi}
\usepackage{natbib}
\usepackage{cleveref}

\input{symbol.tex}

\usepackage{color}

\theoremstyle{plain}
\newtheorem{theorem}{Theorem}
\newtheorem{lemma}{Lemma}
\newtheorem{proposition}{Proposition}
\newtheorem{corollary}{Corollary}

\theoremstyle{definition}
\newtheorem{definition}{Definition}

\theoremstyle{remark}

\newcommand{\email}[1]{\texttt{#1}}
\newcommand{\orcid}[2]{
    \href{https://orcid.org/#1}
    {\includegraphics[scale=.06]{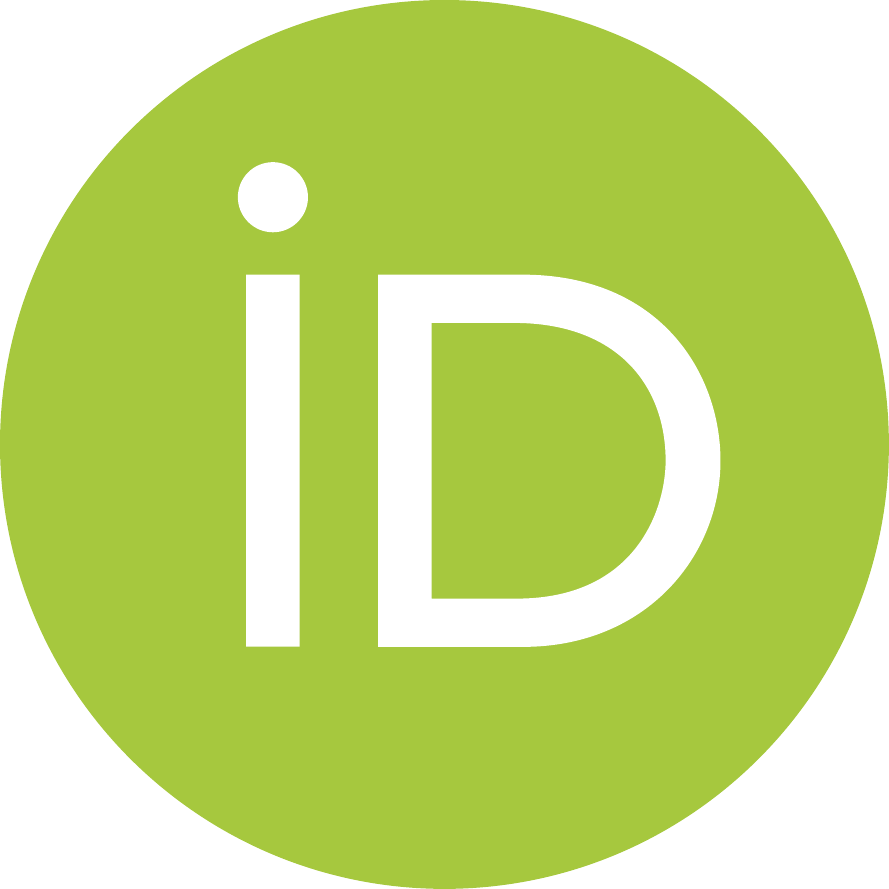} \hspace{1mm} #2}
}

\newcommand{\rootnode}{\mathrm{root}}
\newcommand{\leafnode}{\mathrm{leaf}}

\title{Extended formulations via decision diagrams}

\author{
    \orcid{0000-0000-0000-0000}{Yuta Kurokawa} \\
    Kyushu University\\
    \email{ie19004r@outlook.jp} \\
    \And
    \orcid{0000-0003-2277-6750}{Ryotaro Mitsuboshi} \\
    Kyushu University/RIKEN AIP\\
    \email{ryotaro.mitsuboshi@inf.kyushu-u.ac.jp} \\
    \And
    \orcid{0000-0001-8393-9002}{Haruki Hamasaki} \\
    Kyushu University\\
    \email{hamasaki.haruki.897@s.kyushu-u.ac.jp} \\
    \And
    \orcid{0000-0002-1536-1269}{Kohei Hatano} \\
    Kyushu University/RIKEN AIP\\
    \email{hatano@inf.kyushu-u.ac.jp} \\
    \And
    \orcid{0000-0001-9542-2553}{Eiji Takimoto} \\
    Kyushu University\\
    \email{eiji@inf.kyushu-u.ac.jp} \\
    \And
    \orcid{0000-0000-0000-0000}{Holakou Rahmanian} \\
    Amazon \\
    \email{holakou@amazon.com} \\
}

\begin{document}

\maketitle

\begin{abstract}
    We propose a general algorithm of constructing an extended formulation
    for any given set of linear constraints with integer coefficients.
    Our algorithm consists of two phases: first construct a decision diagram
    $(V,E)$ that somehow represents a given $m \times n$ constraint matrix,
    and then build an equivalent set of $|E|$ linear constraints over
    $n+|V|$ variables.
    That is, the size of the resultant extended formulation depends not
    explicitly on the number $m$ of the original constraints,
    but on its decision diagram representation.
    Therefore, we may significantly reduce the computation
    time and space for optimization problems with integer constraint matrices
    by solving them under the extended formulations, especially when we
    obtain concise decision diagram representations for the matrices.
    Then, we consider the $1$-norm regularized soft margin
    optimization over the binary instance space $\{0,1\}^n$, 
    a standard optimization problem in the machine learning literature. 
    This problem is motivating 
    since the naive application of 
    our extended formulation produces decision diagrams of size $\Omega(m)$. 
    For this problem, 
    we give a modified formulation 
    which works in practice and efficient algorithms 
    whose time complexity depends on the size of the diagrams.
    This problem can be formulated as a linear programming problem with
    $m$ constraints with $\{-1,0,1\}$-valued coefficients over $n$ variables,
    where $m$ is the size of the given sample. 
    We demonstrate the effectiveness of our extended formulations for
    mixed integer programming 
    and the $1$-norm regularized soft margin optimization
    tasks over synthetic and real datasets.

\end{abstract}

\input{intro.tex}

\input{related.tex}

\input{prelim.tex}

\input{lin_const.tex}
\input{lin_const_integer.tex}
\input{softmargin.tex}
\input{softmargin_lp.tex}
\input{erlp_on_zdd.tex}
\input{nzdd_construction.tex}

\input{experiments.tex}

\input{conclusion.tex}
\input{acknowledgement.tex}

\bibliographystyle{abbrvnat}
\bibliography{hatano,mitsuboshi}

\newpage
\appendix
\input{appendix/appendix_nzdd_construction.tex}
\input{appendix/appendix_experiments.tex}

\end{document}

%% file: symbol.tex
\newcommand{\sign}{\mathrm{sign}}
\newcommand{\dom}{\mathcal{X}} 
\newcommand{\calP}{\mathcal{P}}
\newcommand{\Real}{\mathbb{R}} 

\newcommand{\vecx}{\bm{x}}

\newcommand{\vecw}{\bm{w}}
\newcommand{\vecz}{\bm{z}}

\newcommand{\veca}{\bm{a}}
\newcommand{\vecb}{\bm{b}}
\newcommand{\vecs}{\bm{s}}
\newcommand{\vecd}{\bm{d}}

\newcommand{\vecxi}{\bm{\xi}}

\newcommand{\idx}{\mathrm{idx}}

\newcommand{\vecbeta}{\bm{\beta}}

\newcommand{\vecc}{\bm{c}}
\newcommand{\matA}{\bm{A}}

\newcommand{\eps}{\varepsilon}

\def\mytt#1{%
{\tt #1}}

%% file: intro.tex
\section{Introduction}
\label{sec:intro}
Large-scale optimization tasks appear in many areas such as
 machine learning, operations research, and engineering.
Time/memory-efficient optimization techniques are more in demand than ever.
Various approaches have been proposed to efficiently solve
optimization problems over huge data,
e.g., stochastic gradient descent methods (e.g.,\cite{duchi-etal:jmlr11})
and concurrent computing techniques using GPUs
(e.g.,\cite{raina-etal:icml09}).
Among them, we focus on the ``computation on compressed data'' approach,
where we first compress the given data somehow 
and then employ an algorithm that works directly on the compressed data
(i.e., without decompressing the data) to complete the task, in an
attempt to reduce computation time and/or space.
Algorithms on compressed data are mainly studied 
in string processing
(e.g.,~\cite{goto-etal:jda13,hermelin-etal:stacs09,lifshits:cpm07,lohrey:groups12,rytter:icalp04}),
enumeration of combinatorial objects(e.g.,~\cite{minato:ieice17}),
and combinatorial optimization(e.g.,~\cite{bergman-etal:book16}).
In particular, in the work on combinatorial optimization,
they compress the set of feasible solutions that satisfy given constraints
into a decision diagram so that minimizing a linear objective
can be done by finding the shortest path in the decision diagram.
Although we can find the optimal solution very efficiently when
the size of the decision diagram is small, the method can only be applied to
specific types of discrete optimization problems 
where the feasible solution set is finite, 
and the objective function is linear.

Whereas, we mainly consider a more general form of
discrete/continuous optimization problems that include
linear constraints with integer coefficients:
\begin{equation}\label{prob:lin_const_opt}
    \min_{\vecx \in X \subset \Real^n} f(\vecx)
    \quad \text{s.t.} \quad \matA \vecx \geq \vecb
\end{equation}
for some $\matA \in C^{m \times n}$ and $\vecb \in C^m$,
where $X$ denotes the constraints other than $\matA \vecx \geq \vecb$, and 
$C$ is a finite subset of integers. 
This class of problems includes LP, QP, SDP, and MIP 
with linear constraints of integer coefficients.
So our target problem is fairly general.
Without loss of generality, we assume $m > n$, and we are particularly
interested in the case where $m$ is huge.

In this paper, we propose a pre-processing method that "rewrites" 
integer-valued linear constraints with equivalent but more concise ones.
More precisely, we propose a general algorithm that, when given
an integer-valued constraint matrix
$(\matA,\vecb) \in C^{m \times n} \times C^m$
of an optimization problem (\ref{prob:lin_const_opt}),
produces a matrix
$(\matA',\vecb') \in C^{m' \times (n+n')} \times C^{m'}$
that represents its extended formulation, that is, it holds that
\[
	\exists \vecs \in \Real^{n'},
		\matA' \begin{bmatrix}
			\vecx \\ \vecs
		\end{bmatrix}
		\geq \vecb' \Leftrightarrow \matA \vecx \geq \vecb
\]
for some $n'$ and $m'$, with the hope that
the size of $(\matA',\vecb')$ is much smaller than that of
$(\matA,\vecb)$ even at the cost of adding $n'$ extra variables.
Using the extended formulation, we obtain an equivalent
optimization problem to (\ref{prob:lin_const_opt}):
\begin{equation}\label{prob:zdd_const_opt1}
    \min_{\vecx\in X \subset \Real^n,\vecs\in \Real^{n'}} f(\vecx) 
	\quad \text{s.t.} \quad
	\matA' 
    \begin{bmatrix}
        \vecx \\ \vecs
    \end{bmatrix}
     \geq \vecb'.
\end{equation}
Then, we can apply any existing generic solvers, e.g., MIP/QP/LP solvers if
$f$ is linear or quadratic, to (\ref{prob:zdd_const_opt1}),
combined with our pre-processing method, 
which may significantly
reduce the computation time/space than applying them to the original
problem (\ref{prob:lin_const_opt}).

To obtain a matrix $(\matA',\vecb')$, we first construct
a variant of a decision diagram called a
Non-Deterministic Zero-Suppressed Decision Diagram
(NZDD, for short)~\cite{fujita-etal:tcs20} that somehow represents the 
matrix $(\matA,\vecb)$.
Observing that the constraint $\matA \vecz \geq \vecb$ can be
restated in terms of the NZDD constructed as
``every path length is lower bounded by 0''
for an appropriate edge weighting, we establish the extended formulation
$(\matA',\vecb') \in C^{m' \times (n+n')} \times C^{m'}$
with $m' = |E|$ and $n' = |V|$, where
$V$ and $E$ are the sets of vertices and edges of the NZDD, respectively.
One of the advantages of the result is that 
the size of the resulting optimization problem depends only on the size of
the NZDD and the number $n$ of variables, but
\emph{not} on the number $m$ of the constraints in the original problem. 
Therefore, if the matrix $(\matA,\vecb)$ is well compressed into
a small NZDD, then we obtain an equivalent but concise optimization
problem (\ref{prob:zdd_const_opt1}).

To clarify the differences between our work and previous work 
regarding optimization using decision diagrams, 
we summarize the characteristics of both results in Table~\ref{tab:summary}. 
Notable differences are that 
(i) ours can treat optimization problems 
with any types of variables (discrete, or real),
any types of objectives (including linear ones) 
but with integer coefficients on linear constraints, and 
(ii) ours uses decision diagrams for representing linear constraints 
while previous work uses them for 
representing feasible solutions of particular classes of problems.
So, for particular classes of discrete optimization problems, 
the previous approach would work better 
with specific construction methods for decision diagrams. 
On the other hand, 
ours is suitable for continuous optimization problems or/and  
discrete optimization problems for which efficient construction methods 
for decision diagrams representing feasible solutions are not known.
See the later section for more detailed descriptions of related work.

\begin{table}[h]
    \begin{center}
        \caption{%
            Characteristics of previous work on optimization %
            with decision diagrams (DDs) and ours.%
        }
        \label{tab:summary}
        \begin{tabular}{ccccc} \toprule
                          & coeff. of lin. consts. & variables      & objectives & DDs
            \\ \midrule
            Previous work & any type               & binary/integer & linear     & feasible solutions
            \\
            Ours          & binary/integer         & any type       & any type   & lin. consts.
            \\ \bottomrule
        \end{tabular}
    \end{center}
\end{table}

 Among various linear optimization problems, 
 we consider the $1$-norm regularized soft margin optimization
 as a non-trivial application of our method. 
 This problem is a standard optimization problem in the machine learning literature, 
 categorized as LP, 
 for finding sparse linear classifiers 
 given labeled instances. 
 This problem is motivating and challenging in that 
 it has $n+m$ variables and $m$ linear constraints, 
 so the naive application of our method will not be successful 
 as the size of the NZDD representing the constraints is $\Omega(m)$. 
 For this problem, we propose a modified formulation that suffices to work well in practice, 
 and we show efficient algorithms whose time complexity depends only on the size of the NZDD 
 for the modified problem.

Furthermore, to realize succinct extended formulations,
we propose practical heuristics for constructing NZDDs,
which is our third contribution.
Since it is not known to construct an NZDD of small size,
we first construct a ZDD of minimal size, where the ZDD is a restricted
form of the NZDD representation. To this end, we use a ZDD compression
software called \texttt{zcomp}~\cite{toda:ds13}.
Then, we give rewriting rules for NZDDs that reduce both the numbers of vertices
and edges, and apply them to obtain NZDDs of smaller size of $V$ and $E$.
Although the rules may increase the size of NZDDs
(i.e., the total number of edge labels), the rules seem to work
effectively since reducing $|V|$ and $|E|$ is more important
for our purpose.

Experimental results on synthetic and real data sets show that 
our algorithms improve time/space efficiency significantly, especially when 
(i) $m \gg n$, and 
(ii) the set $C$ of integer coefficients is small, 
e.g., binary, where the datasets tend to have concise NZDD representations.

%% file: related.tex
\section{Related work}

Various computational tasks over compressed strings or texts are investigated 
in algorithms and data mining literature, including, e.g., 
pattern matching over strings and computing edit distances or $q$-grams
~\cite{goto-etal:jda13,hermelin-etal:stacs09,lifshits:cpm07,lohrey:groups12,rytter:icalp04}.
The common assumption is that strings are compressed using the straight-line program, 
which is a class of context-free grammars generating only one string (e.g., LZ77 and +LZ78). 
As notable applications of string compression techniques to data mining and machine learning,
Nishino et al.~\cite{nishino-etal:sdm14} and Tabei et al.~\cite{tabei-etal:kdd16} 
reduce the space complexity of matrix-based computations.
So far, however, string compression-based approaches do not seem to be useful 
for representing linear constraints.

Decision diagrams are used in the enumeration of combinatorial objects, discrete optimization and so on.
In short, a decision diagram is a directed acyclic graph with a root and a leaf, 
representing a subset family of some finite ground set $\Sigma$ or, equivalently, a boolean function. 
Each root-to-leaf path represents a set in the set family. 
The Binary Decision Diagram (BDD)\cite{bryant:ieee-tc86,knuth:book11} and 
its variant, the Zero-Suppressed Binary Decision Diagram (ZDD)\cite{knuth:book11,minato:dac93},  
are popular in the literature. 
These support various set operations (such as intersection and union) in efficient ways.
Thanks to the DAG structure, 
linear optimization problems over combinatorial sets $X\subset \{0,1\}^n$ can be reduced to 
shortest/longest path problems over the diagrams representing $X$. 
This reduction is used to solve the exact optimization of NP-hard combinatorial problems
(see, e.g., \cite{bergman-etal:book16,bergman-etal:cpaior11,castro-etal:informs19,inoue-etal:ieee-sg14,morrison-etal:informs16}) 
and enumeration tasks~\cite{minato:ieice17,minato-uno:sdm10,minato-etal:pakdd08}.
Among work on decision diagrams, 
the work of Fujita et al.\cite{fujita-etal:tcs20} would be closest to ours. 
They propose a variant of ZDD called the Non-deterministic ZDD (NZDD) to represent labeled instances and 
show how to emulate the boosting algorithm AdaBoost$^*$\cite{ratsch-warmuth:jmlr05}, 
a variant of AdaBoost\cite{freund-schapire:jcss97} that maximizes the margin, 
over NZDDs. 
We follow their NZDD-based representation of the data. But our work is different from Fujita et al. 
in that, they propose specific algorithms running over NZDDs, 
whereas our work presents extended formulations based on NZDDs, which could be used with various algorithms.

The notion of extended formulation arises in combinatorial optimization 
(e.g., \cite{conforti-etal:4or10,yannakakis:jcss91}). 
The idea is to re-formulate a combinatorial optimization with an equivalent different form, 
so that the size of the problem is reduced. 
For example, 
a typical NP-hard combinatorial optimization problem 
has an integer programming formulation of exponential size.
Then a good extended formulation should have a smaller size than the exponential.
Typical work on extended formulation focuses on some characterization of the problem 
to obtain succinct formulations (see, e.g., \cite{fiorini-etal:mp21}). 
Our work is different from these in that we focus on the redundancy of the data and 
try to obtain succinct extended formulations for optimization problems described with data.

%% file: prelim.tex
\section{Preliminaries}

The non-deterministic Zero-suppressed Decision Diagram (NZDD)~\cite{fujita-etal:tcs20} 
is a variant of the Zero-suppressed Decision Diagram(ZDD)~\cite{minato:dac93,knuth:book11} , 
representing subsets of some finite ground set $\Sigma$.
More formally, NZDD is defined as follows. 
\begin{definition}[NZDD]
An NZDD $G$ is a tuple $G=(V,E,\Sigma,\Phi)$, 
where $(V,E)$ is a directed acyclic graph 
($V$ and $E$ are the sets of nodes and edges, respectively) 
with a single root with no-incoming edges and a leaf with no outgoing edges, 
$\Sigma$ is the ground set, and 
$\Phi:E \to 2^\Sigma$ is a function assigning each edge $e$ a subset $\Phi(e)$ of $\Sigma$. 
More precisely, we allow $(V,E)$ to be a multigraph, i.e., 
two nodes can be connected with more than one edge.

Furthermore, an NZDD $G$ satisfies the following additional conditions. 
Let $\calP_G$ be the set of paths in $G$ starting from the root to the leaf, 
where each path $P\in \calP_G$ is represented as a subset of $E$, and for any path $P\in \calP_G$, 
we abuse the notation and let $\Phi(P)=\cup_{e\in P}\Phi(e)$.
\end{definition}
\begin{enumerate}
    \item For any path $P\in \calP_G$ and any edges $e, e'\in P$, $\Phi(e) \cap \Phi(e')=\emptyset$. 
    That is, for any path $P$, an element $a \in \Sigma$ appears at most once in $P$. 
    \item For any paths $P, P'\in \calP_G$, $\Phi(P) \neq \Phi(P')$. 
    Thus, each path $P$ represents a different subset of $\Sigma$.
\end{enumerate}
Then, an NZDD $G$ naturally corresponds to a subset family of $\Sigma$. 
Formally, let $L(G)=\{\Phi(P) \mid P \in \calP_G\}$. 
Figure~\ref{fig:NZDD} illustrates an NZDD 
representing a subset family $\{\{a, b, c\},\{b\}, \{b, c, d\}, \{c,d\}\}$.
\begin{figure}[t]
    \begin{center}
        \includegraphics[scale=.7,keepaspectratio]{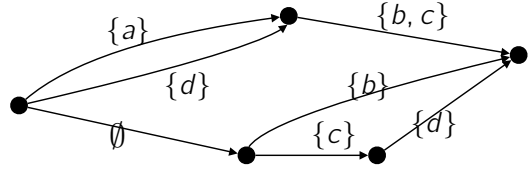}
        \caption{An NZDD representing %
        $\{ \{ a, b, c \}, \{b\}, \{b, c, d\}, \{c, d\}\}$. }
        \label{fig:NZDD}
    \end{center}
    \end{figure}

A ZDD~\cite{minato:dac93,knuth:book11} can be viewed as a special form of NZDD $G=(V,E,\Sigma,\Phi)$ 
satisfying the following properties:
(i) For each edge $e\in E$, $\Phi(e)=\{a\}$ for some $a\in \Sigma$ or $\Phi(e)=\emptyset$.
(ii) Each internal node has at most two outgoing edges. If there are two edges, one is labeled with $\{a\}$ for some $a\in \Sigma$ and the other is labeled with $\emptyset$.
(iii) There is a total order over $\Sigma$ such that, for any path $P\in \calP_G$ and for any $e, e'\in P$ labeled with 
singletons $\{a\}$ and $\{a'\}$ respectively,
if $e$ is an ancestor of $e'$, $a$ precedes $a'$ in the order. 

We believe that that constructing a minimal NZDD 
for a given subset family is
NP-hard since closely related problems are NP-hard. 
For example, 
constructing a minimal ZDD (over all orderings of $\Sigma$) is known to be NP-hard~\cite{knuth:book11}, and 
construction of a minimal NFA which is equivalent to a given DFA is P-space hard~\cite{jiang-ravikumar:sicomp93}.
On the other hand, there is a practical construction algorithm of ZDDs given a subset family and a fixed order over $\Sigma$
using multi-key quicksort~\cite{toda:ds13}.

%% file: lin_const.tex
\section{NZDDs for linear constraints with binary coefficients}

In this section, 
we show an NZDD representation for linear constraints in problem (\ref{prob:lin_const_opt}) 
when linear constraints have $\{0,1\}$-valued coefficients, that is, $C=\{0,1\}$. 
We will discuss its extensions to integer coefficients in the later section.
Let $\veca_i\in \{0,1\}^n$ be the vector corresponding to the $i$-th row of 
the matrix $A\in \{0,1\}^{m\times n}$ (for $i\in [m]$). 
For $\vecx\in \{0,1\}^n$, let $\idx(\vecx)=\{j \in [n] \mid x_j \neq 0\}$, 
i.e., the set of indices of nonzero components of $\vecx$.
Then, we define $I=\{\idx(\vecc_i) \mid \vecc_i =(\veca_i, b_i), i\in[m]\}$. 
Note that $I$ is a subset family of $2^{[n+1]}$. 
Then we assume that we have some NZDD $G=(V,E,[n+1],\Phi)$ representing $I$, that is, $L(G)=I$. 
We will later show how to construct NZDDs.

The following theorem shows the equivalence between 
the original problem (\ref{prob:lin_const_opt}) and a problem described with the NZDD $G$.
\begin{theorem}
\label{theo:main}
    Let $G=(V,E,[n+1],\Phi)$ be an NZDD such that $L(G)=I$. 
    Then the following optimization problem is equivalent to problem (\ref{prob:lin_const_opt}):
    \begin{align}\label{prob:zdd_const_opt}
        \min_{
            \vecx\in X \subset \Real^n,
            \vecs \in \Real^{|V|}
        } & f(\vecx) \\ \nonumber
        \text{s.t.} \quad & s_{e.u} + \sum_{j\in \Phi(e)}x'_j\geq  s_{e.v},
         \quad \forall e \in E,\\ \nonumber
         & s_{\rootnode}=0, ~s_{\leafnode}=0, \\ \nonumber
         & \vecx' = (\vecx, -1),
    \end{align}
    where $e.u$ and $e.v$ are nodes that the edge $e$ is directed from and to, respectively.
\end{theorem} 
Before going through the proof, 
let us explain some intuition on problem (\ref{prob:zdd_const_opt}).
Intuitively, each linear constraint in (\ref{prob:lin_const_opt}) is encoded as 
a path from the root to the leaf in the NZDD $G$, 
and a new variable $s_v$ for each node $v$ represents 
a lower bound of the length of the shortest path 
from the root to $v$. The inequalities in (\ref{prob:zdd_const_opt}) 
reflect the structure of the standard dynamic programming of Dijkstra, 
so that all inequalities are satisfied if and only if the length of all paths is larger than zero. 
In Figure~\ref{fig:extended_ex}, we show an illustration of the extended formulation.

\begin{proof}
    Let $\vecx_\star$ and $(\hat{\vecx}',\hat{\vecs})$ be 
    the optimal solutions of problems (\ref{prob:lin_const_opt})
    and (\ref{prob:zdd_const_opt}), respectively.
    It suffices to show that 
    each optimal solution can 
    construct a feasible solution of the other problem. 

    Let $\hat{\vecx}$ be the vector consisting of the first $n$ components of %
    $\hat{\vecx}'$. 
    For each constraint $\veca_i^\top \vecx \geq b_i$ ($i \in [m]$) in problem (\ref{prob:lin_const_opt}),
    there exists the corresponding path $P_i \in \calP_G$. 
    By repeatedly applying the first constraint in (\ref{prob:zdd_const_opt} along the path $P_i$, we have 
    $\sum_{e\in P_i}\sum_{j\in \Phi(e)}\hat{z}_j' \geq \hat{s}_{\leafnode}=0$. 
    Further, since $\Phi(P_i)$ represents the set of indices of nonzero components of $\vecc_i$, 
    $
        \sum_{e\in P_i} \sum_{j\in \Phi(e)} \hat{z}_j'
        = \vecc_i^\top \hat{\vecx}' = \veca_i^\top \hat{\vecx} - b_i
    $. 
    By combining these inequalities, we have $\veca_i^\top \hat{\vecx}-b_i\ge 0$. 
    This implies that $\hat{\vecx}$ is a feasible solution of (\ref{prob:lin_const_opt})  and thus 
    $f(\vecx_\star)\leq f(\hat{\vecx})$.

    Let $\vecx_\star'=(\vecx_\star,-1)$. Assuming a topological order on $V$(from the root to the leaf), 
    we define $s_{\star,\rootnode}=s_{\star,\leafnode}=0$ and $s_{\star,v}=\min_{e\in E, e.v=v}s_{\star,e.u} + \sum_{j\in \Phi(e)}z_{\star,j}'$ for each $v\in V\setminus \{\rootnode, \leafnode\}$.
    Then, we have, for each $e \in E$ s.t. $e.v \neq \leafnode$,  
    $s_{\star,e.v}\leq s_{\star,e.u} + \sum_{j\in \Phi(e)}z_{\star,j}'$ by definition. 
    Now, $\min_{e\in E, e.v=\leafnode}s_{\star,e.u} + \sum_{j\in \Phi(e)}z_{\star,j}'$ is achieved by a path $P\in \calP_G$ corresponding to 
    $\arg\min_{i\in [m]}\veca_i^\top \vecx_\star -b_i$, which is $\geq 0$ since $\vecx_\star$ is feasible w.r.t. 
    (\ref{prob:lin_const_opt}).
    Therefore, $
        s_{\star,e.v} \leq s_{\star,e.u} + \sum_{j\in \Phi(e)} z_{\star,j}'
    $ for $e\in E$ s.t. $e.v=\leafnode$ as well. 
    Thus, $(\vecx_\star', s_\star)$ is a feasible solution of (\ref{prob:zdd_const_opt}) and $f(\hat{\vecx}) \leq f(\vecx_\star)$.
\end{proof}

\begin{figure}[t]
    \begin{center}
    \includegraphics[scale=.42,keepaspectratio]{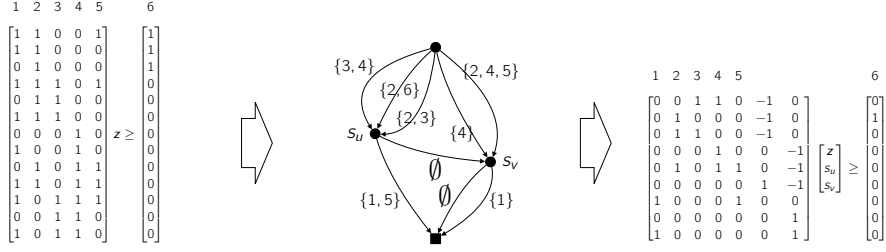}
    \caption{%
        An illustration of the extended formulation. %
        Left: Original constraints as in~(\ref{prob:lin_const_opt}). %
        Middle: A NZDD representation of the left constraints. %
        Right: The matrix form~(\ref{prob:zdd_const_opt}) %
        of the middle diagram without constant terms. %
        This example reduces the $13$ constraints to $9$ constraints %
        by adding $2$ variables. %
    }
    \label{fig:extended_ex}
    \end{center}
\end{figure}


Given the NZDD $G=(V,E)$, problem (\ref{prob:zdd_const_opt}) contains $n+|V|$ variables and $|E|$ linear constraints, 
where $|V|$ variables are real. In particular, if problem (\ref{prob:lin_const_opt}) is LP or IP, 
then problem (\ref{prob:zdd_const_opt}) is LP, or MIP, respectively.

%% file: lin_const_integer.tex
\section{Extensions to integer coefficients}
We briefly discuss how to extend our NZDD representation of linear constraints to the cases 
where coefficients of linear constraints belong to a finite set $C$ of integers.
There are two ways to do so. 
\begin{description}
    \item[Binary encoding of integers] 
    We assume some encoding of integers in $C$ with $O(\log |C|)$ bits. 
    Then, each bit can be viewed as a binary-valued variable. 
    Each integer coefficient can be also recovered with its binary representation. 
    Under this attempt, the resulting extended formulation 
    has $O(n \log |C| + |V|)$ variables and $O(|E|)$ linear constraints. 

    \item[Extending $\Sigma$]
    Another attempt is to extend the domain $\Sigma$ of an NZDD $G=(V, E,\Sigma,\Phi)$. 
    The extended domain $\Sigma'$ consists of all pairs of integers in $C$ and elements in $\Sigma$.
    Again, integer coefficients are recovered through the new domain $\Sigma'$. 
    The resulting extended formulation 
    has $O(n|C| + |V|)$ variables and $O(|E|)$ linear constraints. 
    While the size of the problem is larger than the binary encoding, 
    its implementation is easy in practice and could be effective for $C$ of small size.
\end{description}

%% file: softmargin.tex
\section{$1$-norm regularized soft margin optimization}


The $1$-norm regularized soft margin optimization is a standard linear programming formulation 
of finding a sparse linear classifier with large margin 
(see, e.g., \cite{demiriz-etal:ml02,warmuth-etal:nips07,warmuth-etal:alt08}).
We are given a sequence $S$ of 
labeled instances $S=((\vecx_1,y_1),\dots,(\vecx_m,y_m)) \in (\dom \times \{-1,1\})^m$,
where $\dom \subset \Real^n$ is the set of instances. 
For clarity, we assume that the domain is the set of binary vectors, 
i.e., $\dom=\{0,1\}^n$. 
This is common in many applications
when we employ the bag-of-words representation of instances. 
Given a parameter $\nu \in (0,1]$ and a sequence $S$ of labeled instances, 
the $1$-norm regularized soft margin optimization is defined as follows
\footnote{For the sake of simplicity, we show a restricted version of the formulation. 
We can easily extend the formulation with positive and negative weights
by considering positive and negative weights $w_{j}^+$ and $w_j^-$ instead of $w_j$ ad replacing $w_j$ with $w_j^+ - w_j^-$.}:
\begin{align}\label{prob:softmargin_primal}
    \max_{
        \rho, \vecw, b, \vecxi
    }& \quad \rho 
    - \frac{1}{\nu m}\sum_{i=1}^m \xi_i \\ \nonumber
    \text{s.t.}& \quad y_i(\vecw^\top \vecx_i -b)\geq \rho -\xi_i \quad \forall i=1,\dots,m,\\ \nonumber
    & \quad \sum_{j}w_j + b=1, \vecw \in \Real^n_+, b \geq 0, \vecxi \in \Real^m_+.
\end{align}
For the parameter $\nu\in (0,1]$ and 
the optimal solution $(\rho^\star,\vecw^\star,b^\star,\vecxi^\star)$, 
by a duality argument, 
it can be verified that there are at least $(1-\nu)m$ instances 
that has margin larger than $\rho^\star$, 
i.e., $
    y_i (\vecw^{\star\top}  \vecx_i -b^\star)\geq \rho^\star
$~\cite{demiriz-etal:ml02}.

We formulate a variant of the $1$-norm reguralized soft margin optimization based on NZDDs and 
propose efficient algorithms.
Our generic re-formulation of (\ref{prob:zdd_const_opt}) can be applied 
to the soft margin problem (\ref{prob:softmargin_primal}) as well. 
However, a direct application is not successful since problem (\ref{prob:softmargin_primal}) contains
$O(m)$ slack variables $\vecxi$ and each $\xi_i$ appears only once in $m$ linear constraints. 
That implies a resulting NZDD contains $\Omega(m)$ edges. 
Therefore, we are motivated to formulate a soft margin optimization for which a succinct NZDD representation exists.

Our basic idea is as follows: Suppose that we have some NZDD $G=(V,E,\Phi,\Sigma)$ such that 
each path $P$ corresponds to a constraint $y_i (\vecw^\top \vecx_i +b)\geq \rho$. 
Our idea is to introduce a slack variable $\beta_e$ on each edge $e$ along the path $P$, 
instead of using a slack variable $\xi_i$ for each instance $i\in [m]$.

Given some NZDD $G=(V,E,\Sigma,\Phi)$,
\begin{align}\label{prob:zdd_softmargin_primal}
    \max_{\rho, b, \vecw, \bm{\beta}} \quad & \rho 
    - \frac{1}{\nu m}\sum_{i=1}^m \sum_{e\in P_i}\beta_e \\ \nonumber
    \text{s.t.}\; & y_i(\vecw^\top \vecx_i -b)
    \geq \rho -\sum_{e\in P_i}\beta_e, \; \forall i=1,\dots,m,\\ \nonumber
    & \sum_{j=1}^n w_j + b = 1, \quad
    b \geq 0, \bm{w} \in \Real^n_+, \bm{\beta} \in \Real^E_+.
\end{align}

Note that, the sum of the slack variables $\sum_{e\in P_i}\beta_e$ for each instance $\vecx_i$ 
is more restricted than the original slack variable $\xi_i$. This observation implies the following.

\begin{proposition}
    An optimal solution of problem (\ref{prob:zdd_softmargin_primal}) 
    is a feasible solution of problem (\ref{prob:softmargin_primal}).
\end{proposition}

Although problem (\ref{prob:zdd_softmargin_primal}) is a restricted version of (\ref{prob:softmargin_primal}),
we observe that this restriction does not decrease the generalization ability in our experiments, which is shown later.

Now we introduce an equivalent formulation of (\ref{prob:zdd_softmargin_primal}) 
which is fully described with an NZDD. 
To do so, we specify how to construct the input NZDD as follows.

\paragraph{NZDD $G$ representing the sample $S$}
Let $I^+=\{\idx((\vecx_i,1))\in 2^{[n+1]}\mid y_i=1, i \in [m]\}$ and 
$I^-=\{\idx((\vecx_i,1))\in 2^{[n+1]} \mid y_i=-1, i\in [m]\}$ be 
the set families of indices of nonzero components of positive and negative instances,
respectively.
For $I^+$ and $I^-$, let 
$G^+ = (V^+, E^+, [n+1], \Phi^+)$
and $G^- = (V^-, E^-, [n+1], \Phi^-)$, be corresponding NZDDs such that $L(G^+)=I^+$
and $L(G^-)=I^-$, respectively.
Finally, we connect $G^+$ and $G^-$ in parallel; 
We denote the resulting NZDD 
as $G = (V, E, [n+1], \Phi)$, 
where $V=V^+ \cup V^-\cup \{\rootnode, \leafnode\}\setminus \{\leafnode^+,\leafnode^-\}$ and 
$E=E^+ \cup E^- \cup \{(\rootnode, \rootnode^+), (\rootnode, \rootnode^-)\}$, 
such that 
(i) two leaf nodes $\leafnode^+, \leafnode^-$ of %
$G^+$ and $G^-$ are merged into a new leaf node $\leafnode$, 
(ii) there are two edges $e$ from a new root node $\rootnode$ %
to root nodes of $G^+$ and $G^-$ with $\Phi(e)=\emptyset$, and 
(iii) for other edges $e\in E$, %
$\Phi(e)=\Phi^+(e)$ if $e \in E^+$ and $\Phi(e)=\Phi^-(e)$ if $e \in E^-$.
\begin{align}\label{prob:zdd_softmargin_primal2}
    \max_{
        \rho,
        \vecw,
        \vecbeta,
        \vecs
    } & \quad \rho 
    - \frac{1}{\nu m}\sum_{e\in E}m_e\beta_e \\ \nonumber
    \text{s.t.} \quad & 
        s_{e.u} + \sign(e)\sum_{j \in \Phi(e)}w_j + \beta_e \geq s_{e.v},
        \quad \forall e \in E, \\ \nonumber
    & s_{\rootnode}=0,~s_{\leafnode}\geq \rho, \\ \nonumber
    & \sum_{j=1}^{n} w_j -w_{n+1}=1, \\ \nonumber
    & w_j \geq 0, \; i \in [n], w_{n+1}\leq 0, \vecbeta \geq \bm{0},
\end{align}
where $m_e$ is the number of paths going thorough the edge $e$, and 
$\sign(e)$ is defined as $\sign(e)=-1$ if $e\in E^-$ and $\sign(e)=1$, otherwise.
The constants $m_e$ can be computed in time $O(|E|)$ a priori by a dynamic programming over $G$.
Note that the bias term $-b$ correspond to $w_{n+1}$ for notatinal convenience. 
Then, by following the same proof argument of Theorem~\ref{theo:main}, we have the following corollary.
\begin{corollary}
    Problem (\ref{prob:zdd_softmargin_primal2}) is equivalent to problem (\ref{prob:zdd_softmargin_primal}) .
\end{corollary}

Problem (\ref{prob:zdd_softmargin_primal2}) 
has $O(n+|V|+|E|)$ variables and $O(|E|)$ linear constraints, 
whereas the original formulation (\ref{prob:softmargin_primal}) 
has $O(n+m)$ variables and $O(m)$ linear constraints. 
So, with a concise NZDD representation of the sample, we obtain an extended formulation 
whose size is independent of $m$ of the sample size. 
In later subsections, 
we propose two efficient solving methods for problem (\ref{prob:zdd_softmargin_primal2}).

%% file: softmargin_lp.tex
\subsection{Column Generation}
\begin{algorithm}[t]
    \caption{Column Generation}
    \label{alg:cg_zdd}
    \begin{algorithmic}[1]
        \REQUIRE{NZDD $G=(V, E, \Sigma, \Phi)$}.
        \STATE{%
            Let $\vecd_1 \in [0,1]^{|E|}$ be %
            any vector satisfying (\ref{const:d1}), (\ref{const:d2}), %
            and (\ref{const:d3}). Let $J_0=\emptyset$.%
        }
        \FOR{$t=1,2,...$}
            \STATE{%
                Let $
                    j_{t} = \arg \max_{j \in [n+1]}
                    \sign(j) \sum_{e: j \in \Phi(e)} \sign(e) d_e
                $ and let $\hat{\gamma}_{t}$ be its objective value.%
            }
            \STATE{%
                If $\hat{\gamma}_t \leq \gamma_t + \eps$, %
                let $T=t-1$ and break.%
            }
            \STATE{%
                Let $J_{t} = J_{t-1} \cup \{j_t\}$ and %
                update $(\bm{d}_{t+1}, \gamma_{t+1})$ as:
                \begin{alignat}{2}
                    & \min_{\vecd, \gamma} & & \gamma \\ \nonumber
                    & \text{s.t.} \quad & &
                        \sign(j) \sum_{e: j \in \Phi(e)}
                        \sign(e)d_e \leq \gamma,
                        \forall j \in J_t, \\ \nonumber
                    & & & \text{%
                        constraints (\ref{const:d1}), %
                        (\ref{const:d2}), %
                        and (\ref{const:d3})%
                    }.
                \end{alignat}
            }
        \ENDFOR
        \ENSURE{%
            Output the Lagrangian coefficients %
            $(\vecw_T, \vecbeta_T, \rho_T)$ for the subproblem w.r.t. $J_T$.%
        }
    \end{algorithmic}
\end{algorithm}
In this subsection, 
we propose a column generation-based method for solving 
the modified $1$-norm soft margin optimization problem %
(\ref{prob:zdd_softmargin_primal2}). 
Although the size of the extended formulation 
(\ref{prob:zdd_softmargin_primal2}) does not depend on 
the size of linear constraints $m$ of the original problem,
it still depends on $n$, the size of variables. 
The column generation is a standard approach of linear programming 
that tries to reduce either the size of linear constraints/variables 
by solving smaller subproblems.
In our case, we try to avoid problems depending on the size of variables $n$. 

By a standard dual argument of linear programming, 
the equivalent dual problem of (\ref{prob:zdd_softmargin_primal2}) is 
given as follows.
\begin{align}
    \label{prob:zdd_softmargin_dual}
    \min_{\gamma, \vecd} \gamma & \\ \nonumber
    \text{s.t.}
        \quad &
            \sign(j) \sum_{e: j\in\Phi(e)} \sign(e)d_e
            \leq \gamma \quad (j \in [n+1])\\
    \label{const:d1}
        &
            \sum_{e: e.u=u}d_e = \sum_{e: e.v=u} d_e,
            \; \forall u \in V\setminus \{\rootnode, \leafnode\}, \\
    \label{const:d2}
        &
            \sum_{e:e.u=\rootnode}d_e =1,
            \sum_{e:e.v=\leafnode}d_e =1, \\
    \label{const:d3}
        &  0 \leq d_e \leq \frac{m_e}{\nu m} \quad (e \in E),
\end{align}
where $\sign(j)=1$ for $j\in [n]$ and $sign(j)=-1$ for $j=n+1$. 
Here, the dual problem (\ref{prob:zdd_softmargin_primal2}) has 
$|E|+1$ variables and $n$ linear constraints. 
Roughly speaking, this problem is to 
find a vector $\vecd$ that represents 
a ``flow'' from the root to the leaf in the NZDD 
optimizing some objective, where the total flow is $1$. 
The objective is $\gamma$, 
the upper bound of $\sign(j) \sum_{e: j \in \Phi(e)} \sign(e) d_e$ 
for each $j \in[n+1]$.
The column generation-based algorithm is 
given in Algorithm~\ref{alg:cg_zdd}. 
The algorithm repeatedly solves the subproblems 
whose constraints related to $\gamma$ are 
only restricted to a subset $J_t\subseteq [n+1]$.
Then it adds $j_{t+1}$ to $J_t$ (updated as $J_{t+1}$), 
where $j_{t+1}$ corresponds to the constraint 
that violates condition (\ref{const:d1}) the most
with respect to the current solution $(\gamma_{t}, \vecd_{t})$.
It can be shown that the column-generation algorithm finds 
an $\eps$-approximate solution. 
\begin{theorem}
    \label{theo:cg}
    Algorithm \ref{alg:cg_zdd} outputs an $\eps$-approximate solution 
    of (\ref{prob:zdd_softmargin_primal2}).
\end{theorem}

\begin{proof}
    Let $(\vecd^\star, \gamma^\star)$ be 
    an optimal solution of (\ref{prob:zdd_softmargin_dual}) 
    and let $\pi^\star$ and $\pi_T$ be 
    the optimum of (\ref{prob:zdd_softmargin_primal2}) and 
    the primal one of the dual subproblem for $T$, respectively. 
    By the duality, we have $\gamma_{T} = \pi_{T}$ 
    and $\gamma^\star = \pi^\star$. 
    We show that $\pi_T \geq \pi^\star - \eps$. 
    By definition of $T$, for any $j \in [n+1]\setminus J_T$, 
    $
        \sign(j) \sum_{e: j \in \Phi(e)} \sign(e) d_e
        \leq \hat{\gamma}_{T+1} \leq \gamma_{T+1} + \eps
    $.
    Then, $(\vecd_{T+1}, \gamma_{T+1} + \eps)$ is 
    a feasible solution of (\ref{prob:zdd_softmargin_dual}).
    So, $
        \sign(j_{T+1}) \sum_{e: j \in \Phi(e)} \sign(e) d_e
        = \hat{\gamma}_{T+1}\geq \gamma^\star
    $.
    Combining these observations, %
    $\pi_{T} = \gamma_{T}\geq \gamma^\star - \eps = \pi^\star - \eps$ %
    as claimed.
\end{proof}
As for its time complexity analysis,
similar to other column generation techniques, 
we do not have non-trivial iteration bounds.
In the next section, 
we propose another algorithm with theoretical guarantee of an iteration bound.

%% file: erlp_on_zdd.tex
\subsection{Performing ERLPBoost over an NZDD}
We can emulate ERLPBoost~\cite{warmuth-etal:alt08} on an NZDD.
The algorithm is the same as ERLPBoost except for the update rule.
Let $d^0_e = \frac{m_e}{m}$ for all $e \in E$.
In each iteration $t$, the compressed version of ERLPBoost 
solves the following sub-problem:
\begin{align}
    \label{eq:compressed_erlpboost_subproblem}
    \min_{\gamma, \boldsymbol d} \quad
    &
    \gamma
    + \frac 1 \eta \left[
            \sum_{e \in E} d_e \ln \frac{d_e}{d^0_e} - d_e + d^0_e
        \right] \\
    \text{s.t.} \quad & 
    \sign(j) \sum_{e: j \in \Phi(e)} \sign(e) d_e \leq \gamma,
    \quad \forall j \in J_t, \nonumber \\
    & \text{constraints~(\ref{const:d1}), (\ref{const:d2}), and~(\ref{const:d3})}.
    \nonumber
\end{align}
Here, $\eta > 0$ is some parameter.
One can rewrite~(\ref{eq:compressed_erlpboost_subproblem})
in terms of $\bm{d}$ by introducing $\max$ function. 
We denote the resulting objective function as $P^{t}(\bm{d})$.
Algorithm~\ref{alg:compressed_erlpboost} shows ERLPBoost over an NZDD.
\begin{algorithm}[t]
    \caption{ERLPBoost over an NZDD}
    \label{alg:compressed_erlpboost}
    \begin{algorithmic}[1]
        \REQUIRE{NZDD $G=(V,E,\Sigma,\Phi)$}
        \STATE{%
            Set $d^0_e = \frac{m_e}{m}$ for all $e \in E$. %
            Let $J_0=\emptyset$.%
        }
        \FOR{$t=1,2, \dots$}
            \STATE{%
                \label{alg-line:weak-learnability}
                Find a hypothesis $j_t \in [n+1]$ %
                that maximizes the edge w.r.t. $\bm{d}^{t-1}$.%
            }
            \STATE{%
                Set $
                    \delta^t :=
                    \min_{1 \leq q \leq t}
                    P^q(\bm{d}^{q-1}) - P^{t-1}(\bm{d}^{t-1})
                $.%
            }
            \STATE{If $\delta^t \leq \varepsilon / 2$, Set $T=t-1$ and break.}
            \STATE{%
                Compute the minimizer $\bm d^t$ of %
                (\ref{eq:compressed_erlpboost_subproblem}).%
            }
        \ENDFOR
        \STATE{%
            Solve (\ref{prob:zdd_softmargin_dual}) over %
            $J_T = \{j_t\}_{t=1}^T$ %
            to get the optimal weights $\bm w^T$ on hypotheses.%
        }
        \ENSURE{Output $f = \sum_{j \in J_T} w^T_j h_j$.}
    \end{algorithmic}
\end{algorithm}
We can also obtain a similar iteration bound like ERLPBoost.
\begin{theorem}
    \label{thm:convergence_guarantee_of_compressed_erlpboost}
    If $
        \eta = \frac{4}{\varepsilon} 
        \mathop{\rm depth}(G) \max\{1, \ln \frac 1 \nu\}
    $, then Algorithm~\ref{alg:compressed_erlpboost} 
    finds an $\varepsilon$-approximate solution 
    to~(\ref{prob:zdd_softmargin_dual}) in
    $
        T \leq \frac{144}{\varepsilon^2}
        \mathop{\rm depth}(G)^2 \max\left(1, \ln \frac 1 \nu \right)
    $
    iterations.
\end{theorem}
Here, we prove Theorem~\ref{thm:convergence_guarantee_of_compressed_erlpboost}
with weaker assumption; we assume that the weak learner returns a hypothesis
$j_{t} \in \{1, 2, \dots, n\}$ satisfying
\begin{align}
    \label{eq:weak_learner_guarantee}
    \sign(j_t)\sum_{e \in E} \sign(e) d^{t-1}_e \geq g
\end{align}
for some unknown value $g > 0$.
This assumption is similar to the one in ERLPBoost~\cite{warmuth-etal:alt08}. 
Under the above assumption, we state the iteration bound.
\begin{theorem}
    \label{thm:convergence_guarantee_of_compressed_erlpboost_strong}
    If $
        \eta = \frac{4}{\varepsilon} 
        \mathop{\rm depth}(G) \max\{1, \ln \frac 1 \nu\}
    $, then Algorithm~\ref{alg:compressed_erlpboost} 
    outputs a solution whose value of~(\ref{prob:zdd_softmargin_dual})
    is at least $g - \varepsilon$.
    Algorithm~\ref{alg:compressed_erlpboost} runs at most
    $
        T \leq \frac{144}{\varepsilon^2}
        \mathop{\rm depth}(G)^2 \max\left(1, \ln \frac 1 \nu \right)
    $
    iterations.
\end{theorem}
To prove Theorem~\ref{thm:convergence_guarantee_of_compressed_erlpboost_strong},
we give some technical lemmata.
\begin{lemma}
    Let $\mathop{\rm depth}(G) = \max_{P \in \mathcal P_G \subset 2^E} |P|$.
    Let $\bm d \in \Real^{|E|}$ be any feasible solution of 
    (\ref{eq:compressed_erlpboost_subproblem}). 
    Then, $\sum_{e \in E} d_e \leq \mathop{\rm depth}(G)$.
\end{lemma}
\begin{proof}
    Let $G'$ be a layered NZDD of $G$ obtained 
    by adding some redundant nodes. 
    Since this operation only increases the number of edges, 
    $\sum_{e \in E} d_e \leq \sum_{e \in E'} d_e$ holds, 
    where $E'$ is the set of edges of $G'$. 
    let $E_k' = \{ e \in E' \mid e \text{ is an edge at depth } k \}$.
    Then, 
    \begin{align*}
        \sum_{e \in E} d_e 
        \leq \sum_{e \in E'} d_e
        = \sum_{k=1}^K \sum_{e \in E_k'} d_e
        = \sum_{k=1}^K 1 = \mathop{\rm depth}(G),
    \end{align*}
    where at the second equality, 
    we used the fact that for a feasible $\bm d$,
    the total flow at any depth equals to $1$.
\end{proof}
The following lemma shows 
an upper bound of the unnormalized relative entropy 
for a feasible distribution. 
\begin{lemma}
    \label{lem:relative_entropy_bound}
    Let $\bm{d}$ be a feasible solution 
    to~(\ref{eq:compressed_erlpboost_subproblem}) for an NZDD $G$.
    Then, the following inequality holds.
    \begin{align}
        \nonumber
        \sum_{e \in E} d_e \ln \frac{d_e}{d^0_e} - d_e + d^0_e
        \leq \mathop{\rm depth}(G) \ln \frac 1 \nu
    \end{align}
\end{lemma}
\begin{proof}
    Without loss of generality, we can assume that the NZDD is layered. 
    Indeed, if the NZDD is not layered, 
    we can add dummy nodes and dummy edges to make the NZDD layered. 
    Further, we can increase the edges on the NZDD to make $m_e = 1$ 
    for all edge $e \in E$. 
    Figure~\ref{fig:layered_dd} %
    and~\ref{fig:edge_duplication} depict
    these manipulations.
    With these conversions, the initial distribution $\bm{d}^0$ becomes 
    $d^0_e = 1/m$ for all $e \in E$ and the feasible region becomes 
    Let $E_k \subset E$ be the set of edges at depth $k$. 
    \begin{align}
        \label{eq:relative_entropy_bound_proof_1}
        \sum_{e \in E} d_e \ln \frac{d_e}{d^0_e} - d_e + d^0_e
        = \sum_{k=1}^{\mathop{\rm depth}(G)} \left(
            \sum_{e \in E_k}
            d_e \ln \frac{d_e}{d^0_e} - d_e + d^0_e
            \right)
        = \sum_{k=1}^{\mathop{\rm depth}(G)} \sum_{e \in E_k}
            d_e \ln \frac{d_e}{d^0_e},
    \end{align}
    where the last equality holds since in each layer, 
    the total flow equals to one. 
    For each edge $e \in E$ with $m_e > 1$, 
    define a set $\hat{E}(e)$ of duplicated edges 
    of size $|\hat{E}(e)| = m_e$. 
    Further, define 
    \begin{align}
        \forall \hat{e} \in \hat{E}(e), \quad
        \hat{d}_{\hat{e}} = \frac{d_e}{m_e}, \quad
        \hat{d}^{0}_{\hat{e}} = \frac{d^0_e}{m_e} =\frac{1}{\nu m}.
        \nonumber
    \end{align}
    Then, 
    \begin{align}
        \nonumber
        \sum_{\hat{e} \in \hat{E}(e)} \hat{d}_{\hat{e}}
        \ln \frac{\hat{d}_{\hat{e}}}{\hat{d}^0_{\hat{e}}}
        = \sum_{\hat{e} \in \hat{E}(e)} \frac{d_e}{m_e}
        \ln \frac{d_e / m_e}{d^0_e / m_e}
        = d_e \ln \frac{d_e}{d^0_e}
    \end{align}
    holds. Thus, for any feasible solution $\bm{d}$, 
    we can realize the same relative entropy value by $\hat{\bm{d}}$. 
    Therefore, we can rewrite~(\ref{eq:relative_entropy_bound_proof_1}) as
    \begin{align}
        \label{eq:relative_entropy_bound_proof_2}
        \sum_{k=1}^{\mathop{\rm depth}(G)} \sum_{e \in E_k}
            d_e \ln \frac{d_e}{d^0_e}
        = \sum_{k=1}^{\mathop{\rm depth}(G)} 
            \sum_{\hat{e} \in \bigcup_{e \in E} \hat{E}(e)}
            \hat{d}_{\hat{e}}
            \ln \frac{\hat{d}_{\hat{e}}}{\hat{d}^0_{\hat{e}}}.
    \end{align}
    By construction of $\hat{E}(e)$, $\hat{d}_{\hat{e}} \in [0, 1/\nu m]$.
    Therefore, we can bound 
    eq.~(\ref{eq:relative_entropy_bound_proof_2}) by
    \begin{align}
        \nonumber
        \sum_{k=1}^{\mathop{\rm depth}(G)} \sum_{e \in E_k}
            d_e \ln \frac{d_e}{d^0_e}
        \leq
        \sum_{k=1}^{\mathop{\rm depth}(G)} \ln \frac{m}{\nu m}
        = \mathop{\rm depth}(G) \frac 1 \nu.
    \end{align}
\end{proof}
\begin{figure}[t]
    \centering
    \begin{minipage}[t]{0.45\linewidth}
        \centering
        \includegraphics[scale=0.4]{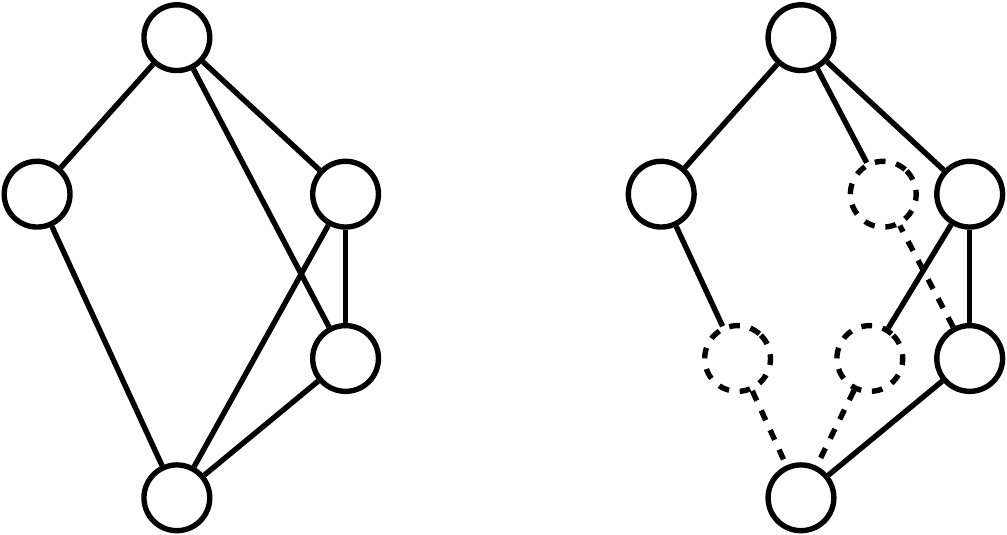}
        \caption{%
            A conversion from an NZDD to a layered NZDD. %
            The added nodes and edges are depicted with dotted lines. %
            Note that this manipulation does not change the depth of %
            the original NZDD. %
        }
        \label{fig:layered_dd}
    \end{minipage}
    \begin{minipage}[t]{0.45\linewidth}
        \centering
        \includegraphics[scale=0.4]{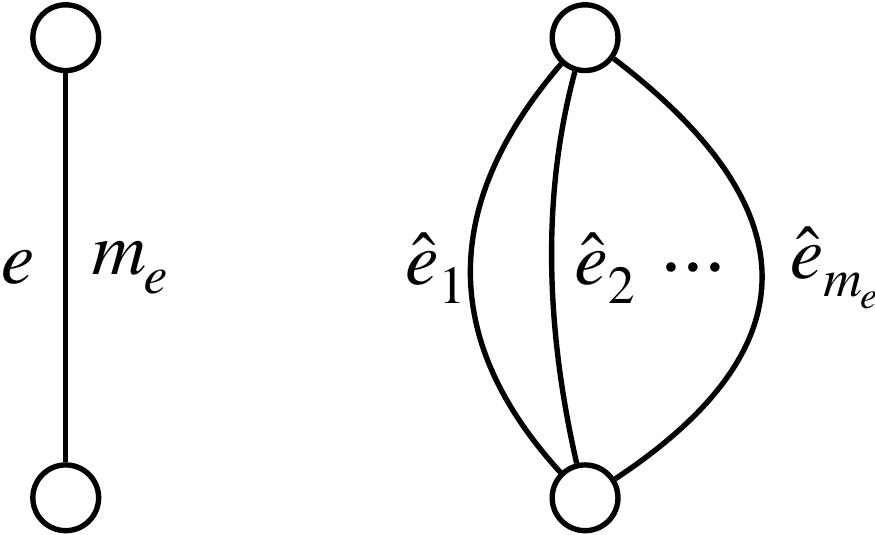}
        \caption{%
            Edge duplication for an edge $e \in E$, i.e., 
            duplicating the edge to the parallel edges of size $m_e$. %
            With this manipulation, $m_{\hat{e}_k}$ becomes one for all %
            duplicated edges $\hat{e}_k \in \hat{E}(e)$. %
        }
        \label{fig:edge_duplication}
    \end{minipage}
\end{figure}
Let $\delta_t = \min_{k \in [t]} P^k(\bm{d}^{k-1}) - P^{t-1}(\bm{d}^{t-1})$ 
be the optimality gap. The following lemma justificates 
the stopping criterion for algorithm~\ref{alg:compressed_erlpboost}.
\begin{lemma}
    \label{lem:erlp_accuracy_guarantee}
    Let $P^{T}_{\mathrm{LP}}$ be an optimal solution 
    of~(\ref{prob:zdd_softmargin_dual})
    over $J_T = \{j_1, j_2, \dots, j_T\}$.
    If $
        \eta \geq \frac{2}{\varepsilon} \mathop{\rm depth}(G) 
        \ln \frac 1 \nu
    $, then $\delta^{T+1} \leq \varepsilon / 2$ implies 
    $g - P^T_{\mathrm{LP}} \leq \varepsilon$, 
    where $g$ is the guarantee of the base learner.
\end{lemma}
Note that if the algorithm 
always finds a hypothesis that maximizes 
the right-hand-side of~(\ref{eq:weak_learner_guarantee}), 
then this lemma guarantees the $\varepsilon$-accuracy to 
the optimal solution of~(\ref{prob:zdd_softmargin_dual}). 
\begin{proof}
    Let $P^T(\bm d)$ be the objective function of 
    (\ref{eq:compressed_erlpboost_subproblem}) over 
    $J_T = \{j_1, j_2, \dots, j_T\}$ obtained 
    by algorithm~\ref{alg:compressed_erlpboost} 
    and let $\bm d^T$ be the optimal feasible solution of it.
    By the choice of $\eta$ and lemma~\ref{lem:relative_entropy_bound}, 
    we get
    \begin{align}
        \nonumber
        \frac 1 \eta \sum_{e \in E} \left[
            d_e \ln \frac{d_e}{d^0_e} - d_e + d^0_e
        \right]
        \leq \frac 1 \eta \mathop{\rm depth}(G) 
            \ln \frac 1 \nu \leq \frac \varepsilon 2.
    \end{align}
    Thus, $P^T(\bm d^T) \leq P^T_{\mathrm LP} + \varepsilon/2$.
    On the other hand, 
    by the assumption on the weak learner, 
    for any feasible solution $\bm{d}^{t-1}$, 
    we obtain a $j_t \in [n+1]$ such that
    \begin{align}
        \nonumber
        g \leq \sign(j_t)\sum_{e \in E} \sign(e) d^{t-1}_e.
    \end{align}
    Using the nonnegativity of
    the unnormalized relative entropy, we get 
    \begin{align}
        g 
        \leq \min_{t \in [T+1]}
        \sign(j_t)\sum_{e \in E} \sign(e) d^{t-1}_e
        \leq \min_{t \in [T+1]} P^t(\bm d^{t-1}).
        \nonumber
    \end{align}
    Subtracting $P^T(\bm d^T)$ from both sides, we get 
    \begin{align*}
        g - P^T(\bm d^T) 
        \leq \min_{t \in [T+1]} P^t(\bm d^{t-1}) - P^t(\bm d^t)
        = \delta^{T+1} \leq \frac \varepsilon 2.
    \end{align*}
    Thus, $\delta_{T+1} \leq \varepsilon/2$ implies 
    $g - P^T_{LP} \leq \varepsilon$. 
\end{proof}
With lemma~\ref{lem:erlp_accuracy_guarantee}, 
We can obtain a similar iteration bound like ERLPBoost.
To see that, we need to derive the dual problem of 
~(\ref{eq:compressed_erlpboost_subproblem}).
By standard calculation, you can verify that the dual problem becomes:
\begin{align}
    \label{eq:compressed_erlpboost_dual}
    \max_{\bm w, \bm s, \bm \beta} \; &
    \frac 1 \eta \sum_{e \in E}
    d^0_e \left( 1 - e^{-\eta A^t_e(\bm{w}, \bm{s}, \bm{\beta})} \right)
    + s_{\rootnode} - s_{\leafnode}
    - \frac{1}{\nu m} \sum_{e \in E} m_e \beta_e \\
    \text{sub. to.} &
    \sum_{j \in J_t} w_j = 1, \bm w \geq \bm 0, \bm \beta \geq \bm 0
    \nonumber \\
    & A^t_e(\bm{w}, \bm{s}, \bm{\beta})
    := \sign(e) \sum_{j \in J_t} \sign(j) I_{[j \in \Phi(e)]} w_j
    + s_{e.u} - s_{e.v} + \beta_e,
    \quad \forall e \in E \nonumber
\end{align}
Let $P^t(\bm d)$, $D^t(\bm w, \bm s, \bm \beta)$ be the objective function 
of the optimization sub-problems 
(\ref{eq:compressed_erlpboost_subproblem}),
(\ref{eq:compressed_erlpboost_dual}), respectively.
Let $\bm d^t$ be the optimal solution of 
(\ref{eq:compressed_erlpboost_subproblem}) at round $t$ and similarly, 
let $(\bm w^t, \bm s^t, \bm \beta^t)$ be the one of 
(\ref{eq:compressed_erlpboost_dual}).
Then, by KKT conditions, the following hold.
\begin{align}
    \label{eq:compressed_erlpboost_kkt}
    \sum_{j \in J_t} \sign(j) w^t_j \left[
        \sum_{e \in E: j \in \Phi(e)} \sign(e) d_e 
    \right]
    = \max_{j \in J_t} \sign(j) \sum_{e \in E: j \in \Phi(e)} \sign(e) d_e
\end{align}

We will prove the following lemma, which corresponds to 
Lemma 2 in~\cite{warmuth-etal:alt08}.
\begin{lemma}
    \label{lem:erlpboost_lemma3}
    If $\eta \geq 1/3$, then 
    \begin{align}
        \nonumber
        P^t(\bm d^t) - P^{t-1}(\bm d^{t-1}) 
        \geq \frac{1}{18 \eta \mathop{\rm depth}(G)} \left[
        P^t(\bm d^{t-1}) - P^{t-1}(\bm d^{t-1})
        \right]^2,
    \end{align}
    where $\mathop{\rm depth}(G)$ denotes the max depth of graph $G$.
\end{lemma}
\begin{proof}
    First of all, we examine the right hand side of the inequality. 
    By definition, $P^t(\bm d^{t-1}) \geq P^{t-1}(\bm d^{t-1})$ and 
    \begin{align}
        P^t(\bm d^{t-1}) - P^{t-1}(\bm d^{t-1})
        = & \sign(j_t)\sum_{e \in E: j_t \in \Phi(e)} \sign(e) d^{t-1}_e
            - \max_{j \in J_{t-1}} \sign(j)
            \sum_{e \in E: j \in \Phi(e)} \sign(e) d^{t-1}_e \nonumber \\
        = & \sign(j_t)\sum_{e \in E: j_t \in \Phi(e)} \sign(e) d^{t-1}_e
            - \sum_{j \in J_{t-1}} \sign(j) w^{t-1}_j
            \sum_{e \in E: j \in \Phi(e)} \sign(e) d^{t-1}_e \nonumber \\
        = & \sum_{e \in E} \sign(e) d^{t-1}_e \left[
                \sign(j_t) I_{[j_t \in \Phi(e)]}
                - \sum_{j \in J_{t-1}} \sign(j)
                I_{[j \in \Phi(e)]}
                w^{t-1}_j
            \right] \nonumber \\
        =: & \sum_{e \in E} d^{t-1}_e x^t_e, \nonumber
    \end{align}
    where the second equality holds from (\ref{eq:compressed_erlpboost_kkt}).

    Now, we will bound $P^t(\bm d^t) - P^{t-1}(\bm d^{t-1})$ from below. 
    For $\alpha \in [0, 1]$, let
    \begin{align}
        \bm w^t(\alpha) := 
        (1-\alpha) \begin{bmatrix} \bm w^{t-1} \\ 0 \end{bmatrix}
        + \alpha \begin{bmatrix} \bm 0 \\ 1 \end{bmatrix}.
    \end{align}
    Since $(\bm w^t, \bm s^t, \bm \beta^t)$ is the optimal solution of 
    (\ref{eq:compressed_erlpboost_subproblem}) at round $t$, 
    $
        D^t(\bm w^t, \bm s^t, \bm \beta^t)
        \geq D^t(\bm w^t(\alpha), \bm s^{t-1}, \bm \beta^{t-1})
    $ 
    holds.
    By strong duality,
    \begin{align}
        P^t(\bm d^t) - P^{t-1}(\bm d^{t-1})
        = & D^t(\bm w^t, \bm s^t, \bm \beta^t)
            - D^{t-1}(\bm w^{t-1}, \bm s^{t-1}, \bm \beta^{t-1})
            \nonumber \\
        \geq & D^t(\bm w^t(\alpha), \bm s^{t-1}, \bm \beta^{t-1})
            - D^{t-1}(\bm w^{t-1}, \bm s^{t-1}, \bm \beta^{t-1})
            \nonumber \\
        = & \frac 1 \eta \sum_{e \in E} d^0_e \left(
            1 - e^{
                - \eta
                A^t_e (\bm{w}^t(\alpha), \bm{s}^{t-1}, \bm{\beta}^{t-1})
            }
        \right)
        - \frac 1 \eta \sum_{e \in E} d^0_e \left(
            1 - e^{
                - \eta
                A^{t-1}_e(\bm{w}^{t-1}, \bm{s}^{t-1}, \bm{\beta}^{t-1})
            }
        \right) \nonumber \\
        = & - \frac 1 \eta \sum_{e \in E} d^0_e \left(
                e^{
                    - \eta A^t_e
                    (\bm{w}^{t}(\alpha), \bm{s}^{t-1}, \bm{\beta}^{t-1})
                }
                - e^{
                    - \eta A^{t-1}_e
                    (\bm{w}^{t-1}, \bm{s}^{t-1}, \bm{\beta}^{t-1})
                }
            \right). \nonumber
    \end{align}
    By KKT conditions, we can write $\bm d^{t-1}$ in terms of the 
    dual variables $(\bm w^{t-1}, \bm s^{t-1}, \bm \beta^{t-1})$:
    \begin{align}
        d^{t-1}_e
        = d^0_e \exp\left[
            - \eta \left(
                \sign(e)
                \sum_{j \in J_{t-1}} \sign(j) I_{[j \in \Phi(e)]} w^{t-1}_j
                + s^{t-1}_u - s^{t-1}_v + \beta^{t-1}_e
            \right)
        \right]
        = d^0_e e^{
            - \eta A^{t-1}_e (\bm{w}^{t-1}, \bm{s}^{t-1}, \bm{\beta}^{t-1})
        }
        \nonumber
    \end{align}
    Therefore, 
    \begin{align*}
        P^t(\bm d^t) - P^{t-1}(\bm d^{t-1})
        \geq - \frac 1 \eta \sum_{e \in E} d^0_e 
            e^{
                -\eta A^{t-1}_e
                (\bm{w}^{t-1}, \bm{s}^{t-1}, \bm{\beta}^{t-1})
            }
            \left( e^{-\eta \alpha x^t_e} - 1 \right)
        = - \frac 1 \eta \sum_{e \in E} d^{t-1}_e 
            \left( e^{-\eta \alpha x^t_e} - 1 \right).
    \end{align*}
    Since $x^t_e \in [-2, +2]$, $\frac{3 \pm x^t_e}{6} \in [0, 1]$ and 
    $\frac{3 + x^t_e}{6} + \frac{3 - x^t_e}{6} = 1$. Thus, using 
    Jensen's inequality, we get
    \begin{align*}
        P^t(\bm d^t) - P^{t-1}(\bm d^{t-1})
        \geq & - \frac 1 \eta \sum_{e \in E} d^{t-1}_e 
            \left(
                \exp \left[
                    \frac{3 + x^t_e}{6} (-3 \eta \alpha)
                    + \frac{3 - x^t_e}{6} (3 \eta \alpha)
                \right]
                - 1
            \right) \\
        \geq & - \frac 1 \eta \sum_{e \in E} d^{t-1}_e 
            \left(
                \frac{3 + x^t_e}{6} e^{-3 \eta \alpha}
                + \frac{3 - x^t_e}{6} e^{3 \eta \alpha}
                - 1
            \right) =: R(\alpha).
    \end{align*}
    The above inequality holds for all $\alpha \in [0, 1]$.
    Here, $R(\alpha)$ is a concave function w.r.t. $\alpha$ so that 
    we can choose the optimal $\alpha \in [0, 1]$.
    By standard calculation, we get that the optimal $\alpha \in \Real$ is 
    \begin{align}
        \label{eq:compressed_erlpboost_optimal_alpha}
        \alpha = \frac{1}{6\eta}
        \ln 
        \frac{\sum_{e \in E} d^{t-1}_e (3 + x^t_e)}
             {\sum_{e \in E} d^{t-1}_e (3 - x^t_e)}
    \end{align}
    Since $x^t_e \leq 2$ for all $e \in E$, 
    $\alpha \leq \frac{1}{6\eta} \ln \frac \ln 5 < \frac{1}{3\eta}$. 
    Thus, $\alpha \leq 1$ holds. 
    On the other hand, 
    $\sum_{e \in E} d^{t-1}_e x^t_e \geq 0$ so that 
    $\alpha \geq \frac{1}{6\eta} \ln 1 = 0$. 
    Therefore, we can use (\ref{eq:compressed_erlpboost_optimal_alpha}) 
    to lower-bound $P^t(\bm d^t) - P^{t-1}(\bm d^{t-1})$.
    \begin{align*}
        P^t(\bm d^t) - P^{t-1}(\bm d^{t-1})
        \geq & - \frac 1 \eta \left[
            \frac 1 3 \sqrt{
                \left( \sum_{e \in E} d^{t-1}_e (3 + x^t_e) \right)
                \left( \sum_{e \in E} d^{t-1}_e (3 - x^t_e) \right)
            } - \sum_{e \in E} d^{t-1}_e
        \right] \\
        = & - \frac 1 \eta \left[
            \frac 1 3 \sqrt{
                9 \left(\sum_{e \in E} d^{t-1}_e\right)^2
                - \left(\sum_{e \in E} d^{t-1}_e x^t_e \right)^2
            } - \sum_{e \in E} d^{t-1}_e
        \right] 
    \end{align*}
    By using the inequality
    \begin{align*}
        \forall a, b \geq 0, 
        3a \geq b \implies
        \frac 1 3 \sqrt{9a^2 - b^2} - a
        \leq - \frac{b^2}{18a},
    \end{align*}
    we get
    \begin{align*}
        P^t(\bm d^t) - P^{t-1}(\bm d^{t-1})
        \geq & \frac{1}{18 \eta}
            \frac{\left(\sum_{e \in E} d^{t-1}_e x^t_e\right)^2}
                 {\sum_{e \in E} d^{t-1}_e} \\
        \geq & \frac{1}{18\eta}
            \frac{\left(\sum_{e \in E} d^{t-1}_e x^t_e\right)^2}
                 {\mathop{\rm depth}(G)} 
        = \frac{1}{18\eta \mathop{\rm depth}(G)}
        \left[ P^t(\bm d^{t-1}) - P^{t-1}(\bm d^{t-1}) \right]^2,
    \end{align*}
    which is the inequality we desire.
\end{proof}
We introduce the following lemma 
to prove the iteration bound of the compressed ERLPBoost. 
\begin{lemma}[\cite{abe+:ieice01}]
    \label{lem:erlpboost_recursion}
    Let $(\delta^t)_{t \in \mathbb{N}} \subset \Real_{\geq 0}$ 
    be a sequence such that 
    \begin{align*}
        \exists c > 0, \forall t \geq 1,
        \delta^t - \delta^{t+1} \geq \frac{(\delta^t)^2}{c}.
    \end{align*}
    Then, the following inequality holds for all $t \geq 1$.
    \begin{align*}
        \delta^t \leq \frac{c}{t - 1 + \frac{c}{\delta^1}}
    \end{align*}
\end{lemma}

\begin{proof}[%
    Proof of %
    Theorem~\ref{thm:convergence_guarantee_of_compressed_erlpboost_strong}%
]
    By definition of $\eta$, $\eta \geq 1/3$ holds.
    Thus, by lemma~\ref{lem:erlpboost_lemma3}, we get 
    \begin{align*}
        P^t(\bm d^t) - P^{t-1}(\bm d^{t-1})
        \geq & \frac{1}{18\eta \mathop{\rm depth}(G)}
        \left[ P^t(\bm d^{t-1}) - P^{t-1}(\bm d^{t-1}) \right]^2 \\
        \geq & \frac{1}{18\eta \mathop{\rm depth}(G)}
        \left[ 
            \min_{1 \leq q \leq t} P^q(\bm d^{q-1})
            - P^{t-1}(\bm d^{t-1}) 
        \right]^2 
        = \frac{(\delta^t)^2}{18\eta \mathop{\rm depth}(G)}.
    \end{align*}
    On the other hand,
    \begin{align*}
        P^t(\bm d^t) - P^{t-1}(\bm d^{t-1})
        = & \left(
            \min_{1 \leq q \leq t} P^{q}(\bm d^{q-1}) 
            - P^{t-1}(\bm d^{t-1}) 
        \right) - \left(
            \min_{1 \leq q \leq t} P^{q}(\bm d^{q-1}) - P^t(\bm d^t) 
        \right) \\
        \leq & \left(
            \min_{1 \leq q \leq t-1} P^{q}(\bm d^{q-1}) 
            - P^{t-1}(\bm d^{t-1}) 
        \right) - \left(
            \min_{1 \leq q \leq t} P^{q}(\bm d^{q-1}) - P^t(\bm d^t) 
        \right) \\
        = & \delta^{t-1} - \delta^t.
    \end{align*}
    Thus, we get the following relation:
    \begin{align*}
        \delta^{t-1} - \delta^t \geq 
        \frac{(\delta^t)^2}{c \eta}, \quad
        \text{where} \quad c = 18 \mathop{\rm depth}(G).
    \end{align*}
    Lemma~\ref{lem:erlpboost_recursion} gives us 
    $\delta^t \leq \frac{c \eta}{t-1 + \frac{c \eta}{\delta^1}}$.
    Rearranging this inequality, we get 
    $T \leq \frac{c\eta}{\delta^T} - \frac{c\eta}{\delta^1} + 1$.
    Since 
    $
        \delta^1 
        = \sum_{e \in E: j_1 \in \Phi(e)} \sign(e) d^0_e 
        \leq \mathop{\rm depth}(G)
    $, we have
    $
        \frac{c\eta}{\delta^1}
        \geq \frac{72}{\varepsilon} \mathop{\rm depth}(G)
        > 1
    $.
    Therefore, the above inequality implies that 
    $T \leq \frac{c \eta}{\delta^T}$.
    As long as the stopping criterion does not satisfy, 
    $\delta^t > \varepsilon / 2$ hold.
    Thus, 
    \begin{align*}
        T 
        \leq & \frac{1}{\delta^T} \cdot 36 \mathop{\rm depth}(G) 
            \frac{2}{\varepsilon} \mathop{\rm depth}(G) 
            \cdot \max\left( 1, \ln \frac 1 \nu \right) \\
        \leq & \frac{144}{\varepsilon^2}
            \mathop{\rm depth}(G)^2 \max\left( 1, \ln \frac 1 \nu \right).
    \end{align*}
\end{proof}

%% file: nzdd_construction.tex
\section{Construction of NZDDs}
\label{sec:construction}

We propose heuristics for constructing NZDDs 
given a subset family $S \subseteq 2^{\Sigma}$.
We use the zcomp \cite{toda:ds13,toda:zcomp}, developed by Toda, 
to compress the subset family $S$ to a ZDD. 
The zcomp is designed based on 
multikey quicksort~\cite{bentley-sedgewick:soda97} 
for sorting strings.
The running time of the zcomp is $O(N \log^2|S|)$, 
where $N$ is an upper bound of the nodes of the output ZDD and $|S|$ is 
the sum of cardinalities of sets in $S$. 
Since $N \leq |S|$, the running time is almost linear in the input. 

A naive application of the zcomp is,
however, not very successful in our experiences.
We observe that the zcomp often produces concise ZDDs compared to inputs. 
But, concise ZDDs do not always imply 
concise representations of linear constraints. 
More precisely, the output ZDDs of the zcomp often contains 
(i) nodes with one incoming edge or 
(ii) nodes with one outgoing edge.
A node $v$ of these types introduces a corresponding variable $s_v$ 
and linear inequalities. 
Specifically, in the case of type (ii), we have 
$s_v \leq \sum_{j\Phi(e)}z'_j + s_{e.u}$ for each $e\in E$ s.t. $e.v=v$, and 
for its child node $v'$ and edge $e'$ between $v$ and $v'$,
$s_{v'} \leq \sum_{j \in \Phi(e')}z'_j + s_v$.
These inequalities are redundant 
since we can obtain equivalent inequalities by concatenating them: 
$s_{v'} \leq  \sum_{j \in \Phi(e')} z'_j + \sum_{j \in \Phi(e)} z'_j + s_{e.u}$ 
for each $e\in E$ s.t. $e.v = v$, where $s_v$ is removed.

Based on the observation above, 
we propose a simple reduction heuristics removing nodes of type (i) and (ii). 
More precisely, given an NZDD $G = (V,E)$, 
the heuristics outputs an NZDD $G' = (V',E')$ such that 
$L(G) = L(G')$ and $G'$ does not contain nodes of type (i) or (ii).
The heuristics can be implemented 
in O($|V'|+|E'|+\sum_{e\in E'}|\Phi(e)|$) time 
by going through nodes of the input NZDD $G$ 
in the topological order from the leaf to the root and 
in the reverse order, respectively. 
The details of the heuristics is given in Appendix.

%% file: experiments.tex
\section{Experiments}
\label{sec:experiments}
We show preliminary experimental results on synthetic and real large data sets
\footnote{Codes are available at \url{https://bitbucket.org/kohei_hatano/codes_extended_formulation_nzdd/}.}. 
We performed mixed integer programming and 
$1$-norm regularized soft margin optimization.
Our experiments are conducted on a server with 2.60 GHz Intel Xeon Gold 6124 CPUs and 314GB memory.
We use Gurobi optimizer 9.01,
a state-of-the-art commercial LP solver.
To obtain NZDD representations of data sets, we apply the procedure described 
in the previous section.
The details of preprocessing of data sets 
and NZDD representations are shown in Appendix.

\subsection{Mixed Integer programming on synthetic datasets}
First, we apply our extended formulation (\ref{prob:lin_const_opt}) to 
mixed integer programming tasks over synthetic data sets.
The problems are defined as the linear optimization with $n$ variables and $m$ linear constraints 
of the form $\matA \vecx \geq \vecb$, where 
(i) each row of $\matA$ has $k$ entries of $1$ and others are $0$s and nonzero entries are chosen randomly
without repetition
(ii) coefficients $a_i$ of linear objective $\sum_{i=1}^n{a_ix_i}$ is chosen from {1,...,100} randomly, and
(iii) first $l$ variables take binary values in $\{0,1\}$ and others take real values in $[0,1]$.
In our experiments, we fix $n=25, k=10$, $l=12$ and $m \in \{4\times 10^5, 8\times 10^5...,20\times10^5\}$.
We apply the Gurobi optimizer directly to the problem
denoted as \mytt{mip} 
and the solver with pre-processing the problem by our extended formulation 
(denoted as \mytt{nzdd\_mip}, respectively.
The results are summarized in Figure~\ref{fig:mip_artificial}. 
Our method consistently improves computation time for these datasets.
This makes sense 
since it can be shown that 
when $m=O(n^k)$ there exists an NZDD of size $O(nk)$ representing the constraint matrix.
In addition, the pre-processing time is within 2 seconds in all cases. 

\begin{figure}[ht]
    \begin{center}
        \includegraphics[scale=.35,keepaspectratio]{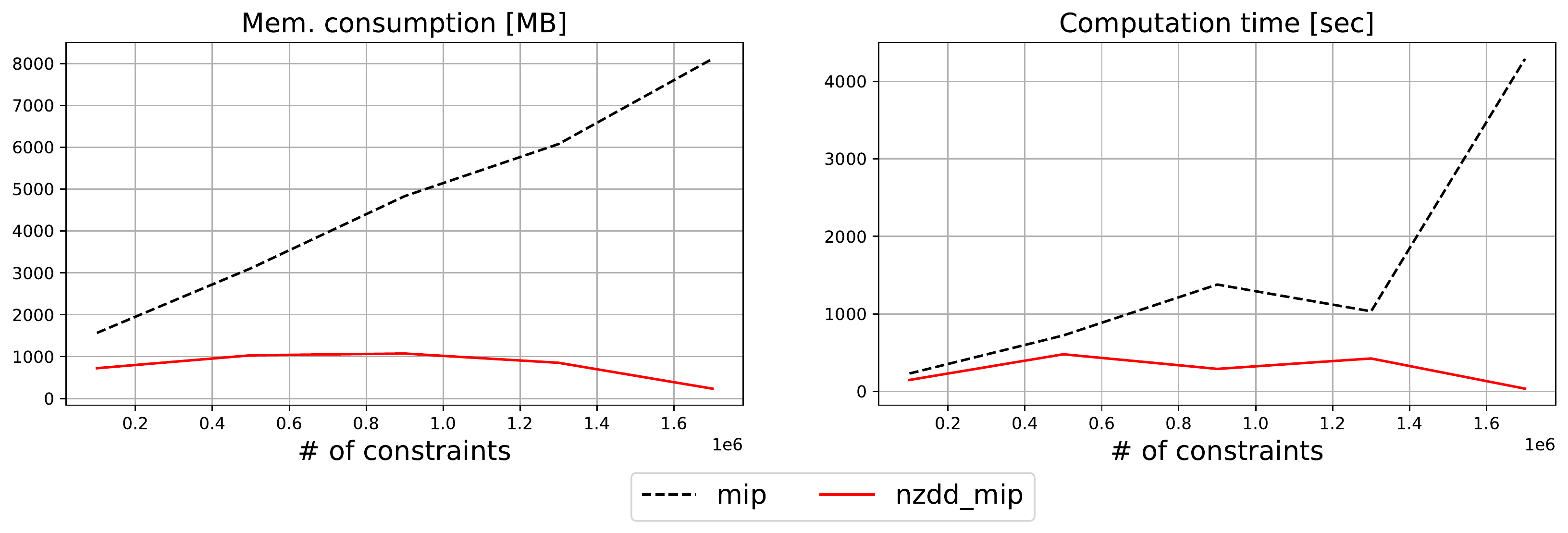}
        \caption{%
            The comparison for synthetic datasets of a MIP problem. %
            The horizontal axis represents the number of constraints of the original problem.
        }
        \label{fig:mip_artificial}
    \end{center}
    
\end{figure}




\subsection{$1$-norm soft margin optimization on real data sets}
Next, we apply our methods on the task of the $1$-norm soft margin optimization. 
This problem is a standard optimization problem in the machine learning literature, 
categorized as LP, for finding sparse linear classifiers given labeled instances. 
We compare the following methods using a naive LP solver.
\begin{enumerate}
    \item a naive LP solver (denoted as \mytt{naive}), 
          \label{item:original-softmargin-start}
    \item LPBoost~(\cite{demiriz-etal:ml02},
          denoted as \mytt{lpb}),
          a column generation-based method,
    \item ERLPBoost~(\cite{warmuth-etal:alt08}, denoted as \mytt{erlp}), 
          a modification of LPBoost with a non-trivial iteration bound.
          \label{item:original-softmargin-end}
    \item a naive LP solver (denoted as \mytt{nzdd\_naive}), 
          \label{item:nzdd-softmargin-start}
    \item Algorithm~\ref{alg:cg_zdd} (denoted as \mytt{nzdd\_lpb}), 
    \item Algorithm~\ref{alg:compressed_erlpboost}
          (denoted as \mytt{nzdd\_erlpb}). 
          \label{item:nzdd-softmargin-end}
\end{enumerate}
Methods~\ref{item:original-softmargin-start}--\ref{item:original-softmargin-end}
solves the original soft margin optimization 
problem~(\ref{prob:softmargin_primal}),
while~\ref{item:nzdd-softmargin-start}--\ref{item:nzdd-softmargin-end} solves
the problem~(\ref{prob:zdd_softmargin_primal2}).
We measure its computation time (CPU time) 
and maximum memory consumption, respectively, 
and compare their averages over parameters. 
Further, we performed $5$-fold cross validation 
to check the test error rates of our methods on real data sets. 
Table~\ref{tab:test_err} shows that 
our formulation (\ref{prob:zdd_softmargin_primal2}) is
competitive with the original soft margin optimization.

We compare methods on some real data sets in the libsvm datasets~\cite{chang-lin:libsvm}
to see the effectiveness of our approach in practice. 
Generally, the datasets contain huge samples ($m$ varies from $3\times10^4$ to $10^7$) 
with a relatively small size of features ($n$ varies from $20$ to $10^5$). 
The features of instances of each dataset is transformed into binary values.
Note that these results exclude NZDD construction times
since the compression takes around $1$ second,
except for the HIGGS dataset (around $13$ seconds). 
Furthermore, the construction time of NZDDs can be neglected in the following reason:
We often need to try multiple choices of the hyperparameters ($\nu$ in our case) and 
solve the optimization problem for each set of choices.
But once we construct an NZDD, 
we can be re-use it for different values of hyperparameters without reconstructing NZDDs.



\begin{figure*}[ht]
    \begin{center}
        \includegraphics[
            width=15cm,height=10cm,keepaspectratio
        ]{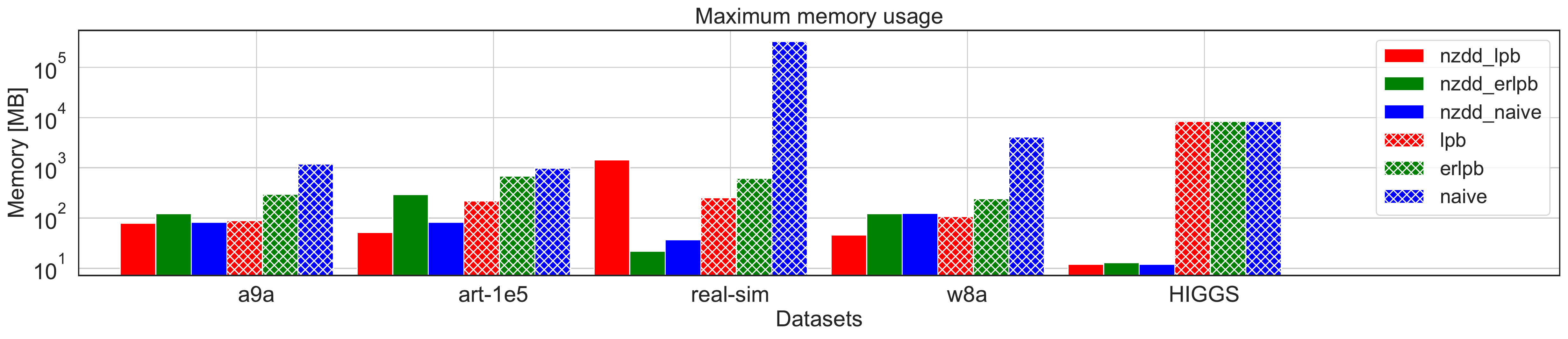}
        \caption{%
            Comparison of maximum memory consumption %
            of the soft margin optimization %
            for real datasets. %
            The y-axes are  plotted in the logarithmic scale.%
        }
        \label{fig:memory_soft_real}
    \end{center}
\end{figure*}

\begin{figure*}[ht]
    \begin{center}
        \includegraphics[
            width=15cm,height=10cm,keepaspectratio
        ]{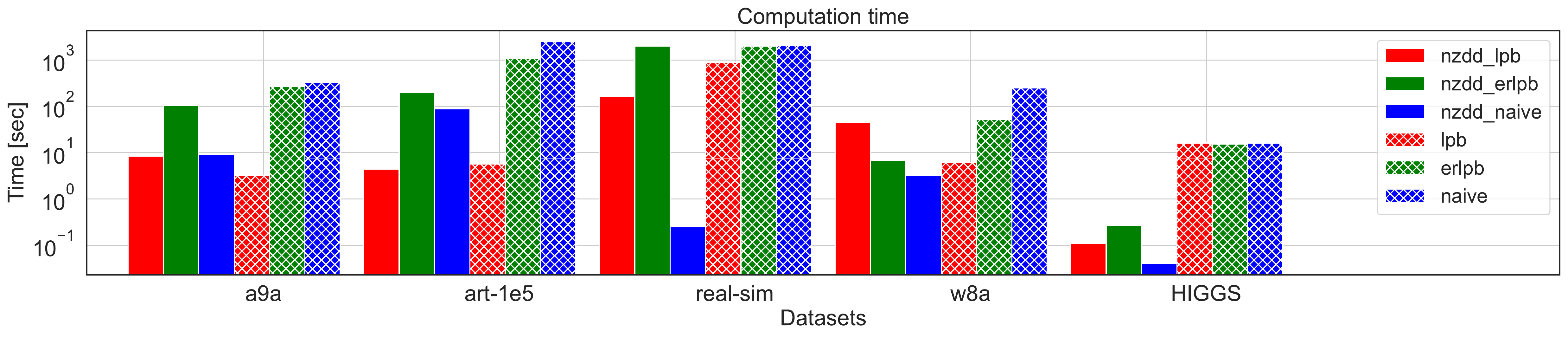}
        \caption{%
            Comparison of running time of the soft margin optimization %
            for real datasets. %
            The y-axes are  plotted in the logarithmic scale.%
        }
        \label{fig:time_soft_real}
    \end{center}
\end{figure*}

\begin{table}[ht]
    \begin{center}
        \caption{Test error rates for real datasets}
        \label{tab:test_err}
        \input{table/softmargin-test-err}

    \end{center}
\end{table}

\paragraph{$1$-norm soft margin optimization on synthetic datasets}
We show experimental results for the $1$-norm soft margin optimization 
on synthetic datasets.
We use a class of synthetic datasets that have small NZDD representations 
when the samples are large. 
First, we choose $m$ instances in $\{0,1\}^n$ 
uniformly randomly without repetition. 
Then we consider the following linear threshold function %
$f(\vecx)=\sign(\sum_{j=1}^k x_j - r+ 1/2)$, 
where $k$ and $r$ are positive integers such that $1\leq r\leq k\leq n$).
That is, $f(x)=1$ if and only if 
at least $r$ of the first $k$ components are $1$. 
Each label of the instance $\vecx\in\{0,1\}^n$ is labeled by $f(\vecx)$. 
It can be easily shown that 
the whole labeled $2^n$ instances of $f$ is 
represented by an NZDD (or ZDD) of size $O(kr)$, 
which is exponentially small w.r.t. the sample size $m=2^n$.  
We fix $n=20$, $k=10$ and $r=5$. 
Then we use $m \in \{1\times 10^5,2 \times 10^5...,10^6\}$.

\begin{figure}[ht]
    \begin{center}
        \includegraphics[
            width=8cm,height=8cm,keepaspectratio
        ]{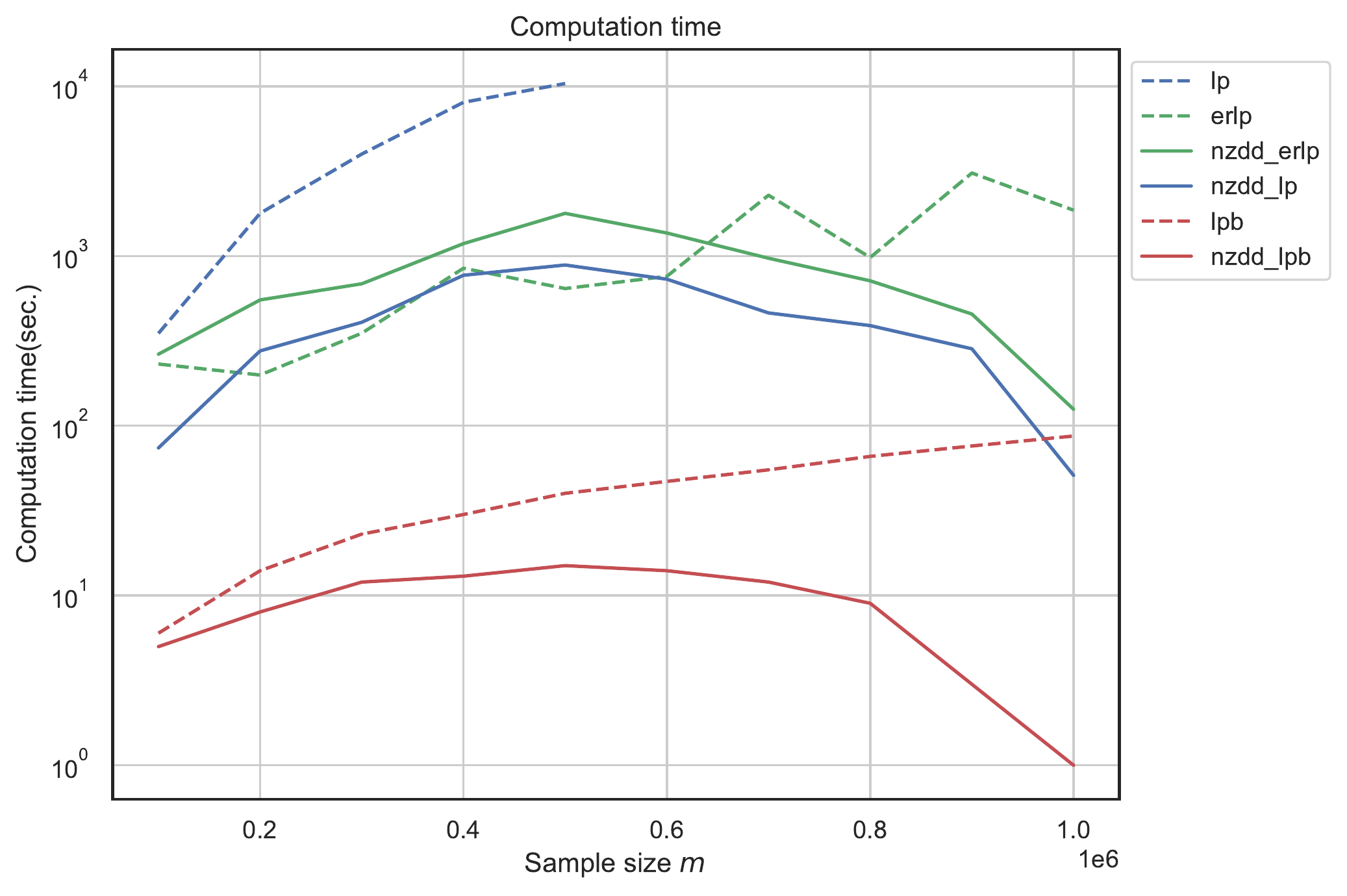}
        \includegraphics[
            width=8cm,height=8cm,keepaspectratio
        ]{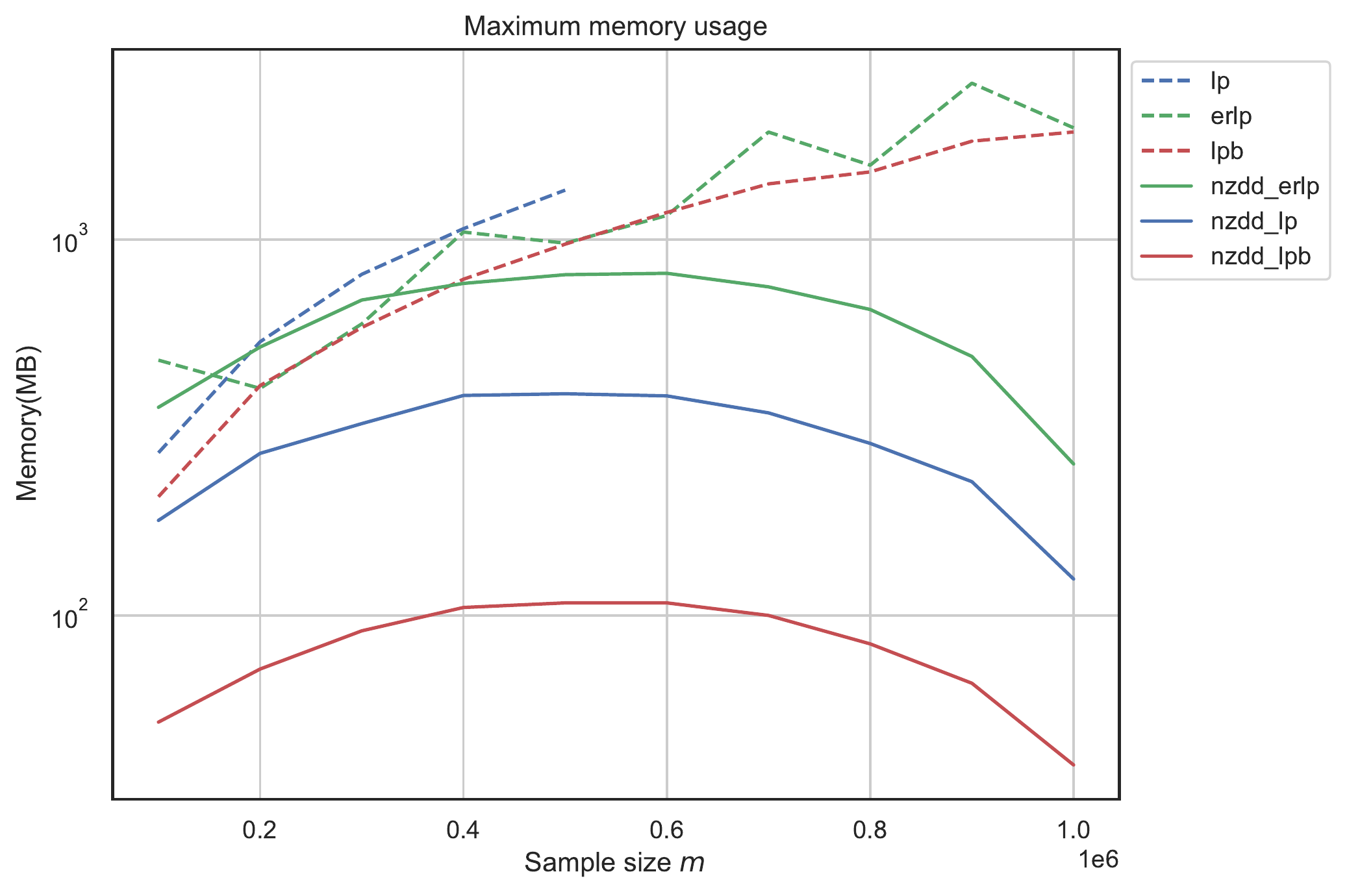}
        \caption{%
            Comparison of computation times %
            for synthetic datasets of the soft margin optimization. %
            The y-axis is shown in the logarithmic scale. %
            Results for \mytt{naive} w.r.t. $m\geq 6 \times 10^5$ are omitted %
            since it takes more than $1$ day $\approx 9\times 10^4$ seconds.%
        }
        \label{fig:softmargin_artificial}
    \end{center}
 \end{figure}

The results are given in Figure~\ref{fig:softmargin_artificial}. 
As expected, our methods are significantly faster than comparators. 
Generally, our methods perform better than the standard counterparts.
In particular, \mytt{nzdd\_lpb} improves efficiency at least $10$ to $100$ times over others. 
Similar results are obtained for maximum memory consumption.

%% file: table/softmargin-test-err.tex
\begin{tabular}{rrrr} \toprule
    Data sets  & \mytt{lpb} & \mytt{nzdd\_naive} & \mytt{nzdd\_erlp} \\ \midrule 
    a9a        & $0.174$ &  $0.159$ & $0.157$ \\
    art-100000 & $0.000$ & $0.0004$ & $0.004$ \\
    real-sim   & $0.179$ &  $0.169$ & $0.532$ \\
    w8a        & $0.030$ &  $0.030$ & $0.029$ \\ \bottomrule
\end{tabular}

%% file: conclusion.tex
\section{Conclusion}
We proposed a generic algorithm of constructing an NZDD-based extended formulation 
for any given set of linear constraints with integer constraints 
as well as specific algorithms for the $1$-norm soft margin optimization and practical heuristics 
for constructing NZDDs. Our algorithms improve time/space efficiency on artificial and real datasets, 
especially when the datasets have concise NZDD representations.

%% file: acknowledgement.tex
\section{Acknowledgements}
We thank Mohammad Amin Mansouri for valuable discussions 
and initial development.
This work was supported by JSPS KAKENHI Grant Numbers
JP19H014174, JP19H04067, \\
JP20H05967, and JP22H03649, respectively.

%% file: appendix/appendix_nzdd_construction.tex
\section{Details of heuristics for constructing NZDDs}

A pseudo-code of the heuristics is given in Algorithm~\ref{alg:reduce}.
Algorithm~\ref{alg:reduce} consists of two phases. 
In the first phase, it traverses nodes in the topological order 
(from the leaf to the root), and for each node $v$ with one incoming edge $e$, 
it contracts $v$ with its parent node $u$ and 
$u$ inherits the edges $e'$ from $v$. Label sets $\Phi(e)$ and $\Phi(e')$ are also merged. 
The first phase can be implemented in $O(|V|+|E|)$ time, 
by using an adjacency list maintaining children of each node 
and lists of label sets for each edge. 
In the second phase, it does a similar procedure to simplify nodes with single outgoing edges. 
To perform the second phase efficiently, we need to re-organize lists of label sets 
before the second phase starts.
This is because the lists of label sets could form DAGS after the first phase ends, 
which makes performing the second phase inefficient. 
Then, the second phase can be implemented in $O(|V'|+|E'|+\sum_{e\in E'}|\Phi(e)|)$ time.

\begin{algorithm}[t]
\caption{Reducing procedure}
\label{alg:reduce}
Input: NZDD $G=(V,E,\Phi,\Sigma)$
\begin{enumerate}
    \item For each $u\in V$ in a topological order (from leaf to root) and for each child node $v$ of $u$,
    \begin{enumerate}   
        \item If indegree of $v$ is one,  
        \begin{enumerate}
            \item for the incoming edge $e$ from $u$ to $v$, 
            each child node $v'$ of $v$ and each outgoing edge $e'$ from $v$ to $v'$,
             add a new edge $e''$ from node $u$ to $v'$ and 
            set $\Phi(e'')=\Phi(e)\cup \Phi(e')$.
        \end{enumerate}
        \item Remove the incoming edge $e$ and all outgoing edges $e'$.
    \end{enumerate}
    \item For each $v\in V$ in a topological order (from root to leaf) and for each parent node $u$ of $v$,
     \begin{enumerate}   
        \item If outdegree of $u$ is one,  
        \begin{enumerate}
                \item for the outgoing edge $e$ of $u$, 
                each parent node $u'$ of $u$ and each outgoing edge $e'$ from $u'$ to $u$,
                 add a new edge $e''$ from node $u'$ to $v$ and 
                set $\Phi(e'')=\Phi(e)\cup \Phi(e')$.
        \end{enumerate}
        \item Remove the outgoing edge $e$ and all incoming edges $e'$.
    \end{enumerate}
        \item Remove all nodes with no incoming and outgoing edges from $V$ and output the resulting DAG $G'=(V',E')$.
\end{enumerate}    
\end{algorithm}

%% file: appendix/appendix_experiments.tex
\section{Details of experiments}
\label{sec:appendix-experiments}
\paragraph{Preprocessing of datasets}
The datasets for the $1$-norm regularized soft margin optimization 
are obtained from the libsvm datasets. 
Some of them contain real-valued features. We convert them to binary ones 
by rounding them using $0.5$ as a threshold.

\paragraph{NZDD construction time and summary of datasets}
Computation times for constructing NZDDs for the MIP synthetic datasets and 
synthetic and real datasets of the $1$norm regularized soft margin optimization
are summarized 
in Table~\ref{tab:nzdd_time_mip}, \ref{tab:nzdd_time_soft_artificial} 
and \ref{tab:nzdd_time_soft_real}, respectively. 
Note that the NZDD construction time is not costly and negligible in general 
because once we construct NZDDs, 
we can re-use those NZDDs for solving optimization problems 
with different objective functions or hyperparameters 
such as $\nu$ in the soft margin optimization.

\begin{table}[ht]
    \centering
    \caption{%
        Computation time (seconds) %
        of NZDD construction for synthetic MIP datasets, %
        described in Table~\ref{tab:nzdd_mip}. %
        All datasets have $n=25$ features %
        and each instance have $k=10$ nonzero components.%
    }
    \label{tab:nzdd_time_mip}
    \begin{tabular}{rrrr}
        \toprule
        $m$ & zcomp & reducing procedure & Total \\
        \midrule
         $4 \times 10^5$ & $0.39$ & $1.02$ & $1.41$ \\
         $8 \times 10^5$ & $0.76$ & $1.38$ & $2.14$ \\
        $12 \times 10^5$ & $1.08$ & $1.41$ & $2.49$ \\
        $16 \times 10^5$ & $1.36$ & $1.10$ & $2.46$ \\
        $20 \times 10^5$ & $1.60$ & $0.33$ & $1.93$ \\
        \bottomrule
    \end{tabular}
\end{table}

\begin{table}[ht]
    \begin{center}
        \caption{%
            Summary of synthetic datasets for MIP. %
            The term ``original'' and ``extended'' mean %
            the original and the extended formulations, respectively.%
        }
        \label{tab:nzdd_mip}
        \input{table/mip-synthetic-info}

    \end{center}
\end{table}
\begin{table}[ht]
    \begin{center}
        \caption{%
            Computation times (seconds) for NZDD construction %
            for the synthetic datasets, %
            described in Table~\ref{tab:summary_soft_artificial}. %
            all data have $n = 20$ features. %
        }
        \label{tab:nzdd_time_soft_artificial}
        \begin{tabular}{rrrr} \toprule
                    $m$ & zcomp & Reducing procedure & Total \\ \midrule 
              $1 \times 10^5$ & $0.08$ &  $0.14$ & $0.22$ \\
              $2 \times 10^5$ & $0.15$ &  $0.14$ & $0.29$ \\
              $3 \times 10^5$ & $0.25$ &  $0.12$ & $0.37$ \\
              $4 \times 10^5$ & $0.31$ &  $0.08$ & $0.39$ \\
              $5 \times 10^5$ & $0.39$ &  $0.05$ & $0.44$ \\
              $6 \times 10^5$ & $0.46$ &  $0.05$ & $0.51$ \\
              $7 \times 10^5$ & $0.53$ &  $0.03$ & $0.56$ \\
              $8 \times 10^5$ & $0.60$ &  $0.02$ & $0.62$ \\
              $9 \times 10^5$ & $0.68$ &  $0.01$ & $0.69$ \\
             $10 \times 10^5$ & $0.72$ &  $0.00$ & $0.72$ \\\bottomrule
        \end{tabular}
    \end{center}
\end{table}

\begin{table}[ht]
    \begin{center}
        \caption{%
            Computation times (seconds) for constructing %
            NZDDs for real datasets of the soft margin optimization.%
        }
        \label{tab:nzdd_time_soft_real}
        \begin{tabular}{rrrr} \toprule
            dataset    & zcomp & Reducing procedure & Total \\ \midrule 
            a9a        &    $0.04$ &    $0.20$ & $0.24$  \\
            art-100000 &    $0.10$ &    $0.43$ & $0.53$  \\
            real-sim   &    $0.24$ &    $0.99$ & $1.23$  \\
            w8a        &    $0.27$ &    $4.20$ & $4.47$  \\ 
            HIGGS      &   $13.13$ &    $0.00$ & $13.13$ \\
            \bottomrule
        \end{tabular}
    \end{center}
\end{table}

Table~\ref{tab:summary_soft} summarizes the size of each problem
for soft margin optimization. 
As this table shows,
the extended formulations (\ref{prob:zdd_softmargin_primal2}) 
have fewer variables and constraints. 
This is not surprising 
since the extended formulation has %
$O(n+|V|+|E|)$ variables and $O(|E|)$ constraints,  
while the original formulation (\ref{prob:softmargin_primal}) has %
$O(n+m)$ variables and $O(m)$ constraints.

\begin{table}[ht]
    \begin{center}
        \caption{%
            Summary of synthetic datasets for the soft margin optimization. %
            The term ``original'' and ``extended'' mean %
            the original and the extended formulations, respectively.
        }
        \label{tab:summary_soft_artificial}
        \input{table/softmargin-synthetic-info}

    \end{center}
\end{table}

\begin{table}[ht]
    \begin{center}
        \caption{
            Summary of sizes of the real datasets %
            of the soft margin optimization. %
            The term ``original'' and ``extended'' mean %
            the original and the extended formulations %
            (\ref{prob:softmargin_primal}) %
            and (\ref{prob:zdd_softmargin_primal2}), respectively.%
        }
         \label{tab:summary_soft}
         \input{table/softmargin-real-info}

    \end{center}
\end{table}

%% file: table/mip-synthetic-info.tex
\begin{tabular}{rrrrrrrr} \toprule
    \multicolumn{2}{c}{Data size} & \multicolumn{2}{c}{NZDD size} & \multicolumn{2}{c}{Variables} & \multicolumn{2}{c}{Constraints}\\
        $n$ &              $m$ &    $|V|$ &  $|E|$    & Original & Extended &    Original & Extended \\ \midrule 
       $25$ & $ 4 \times 10^5$ & $13,321$ & $177,356$ &     $25$ & $13,344$ & $  400,000$ &   $177,356$ \\
       $25$ & $ 8 \times 10^5$ & $17,771$ & $241,488$ &     $25$ & $17,794$ & $  800,000$ &   $241,488$ \\
       $25$ & $12 \times 10^5$ & $19,348$ & $249,749$ &     $25$ & $19,371$ & $1,200,000$ &   $249,749$ \\
       $25$ & $16 \times 10^5$ & $17,819$ & $204,723$ &     $25$ & $17,842$ & $1,600,000$ &   $204,723$ \\
       $25$ & $20 \times 10^5$ & $ 6,161$ & $ 55,555$ &     $25$ & $ 6,184$ & $2,000,000$ &   $ 55,555$ \\ \bottomrule
\end{tabular}

%% file: table/softmargin-synthetic-info.tex
\begin{tabular}{rrrrrrrr} \toprule
    \multicolumn{2}{c}{Data size} & \multicolumn{2}{c}{NZDD size} & \multicolumn{2}{c}{Variables} & \multicolumn{2}{c}{Constraints}\\
     $n$ &        $m$ &     $|V|$ &     $|E|$ &    Original &   Extended &    Original &  Extended \\ \midrule 
    $20$ & $ 1 \times 10^5$ &  $ 4,202$ & $ 55,147$ &  $ 100,021$ &  $ 59,370$ &  $ 200,001$ & $110,297$ \\
    $20$ & $ 2 \times 10^5$ &  $ 6,663$ & $ 82,810$ &  $ 200,021$ &  $ 89,494$ &  $ 400,001$ & $165,623$ \\
    $20$ & $ 3 \times 10^5$ &  $11,022$ & $103,323$ &  $ 300,021$ &  $114,366$ &  $ 600,001$ & $206,649$ \\
    $20$ & $ 4 \times 10^5$ &  $13,596$ & $120,994$ &  $ 400,021$ &  $134,611$ &  $ 800,001$ & $241,991$ \\
    $20$ & $ 5 \times 10^5$ &  $12,565$ & $128,670$ &  $ 500,021$ &  $141,256$ & $1,000,001$ & $257,343$ \\
    $20$ & $ 6 \times 10^5$ &  $13,796$ & $127,139$ &  $ 600,021$ &  $140,956$ & $1,200,001$ & $254,281$ \\
    $20$ & $ 7 \times 10^5$ &  $13,513$ & $114,471$ &  $ 700,021$ &  $128,005$ & $1,400,001$ & $228,945$ \\
    $20$ & $ 8 \times 10^5$ &  $ 9,569$ & $ 95,454$ &  $ 800,021$ &  $105,044$ & $1,600,001$ & $190,911$ \\
    $20$ & $ 9 \times 10^5$ &  $ 6,725$ & $ 75,318$ &  $ 900,021$ &  $ 82,064$ & $1,800,001$ & $150,639$ \\
    $20$ & $10 \times 10^5$ &  $ 3,914$ & $ 40,226$ & $1,000,021$ &  $ 44,161$ & $2,000,001$ & $ 80,455$ \\ \bottomrule
\end{tabular}

%% file: table/softmargin-real-info.tex
\begin{tabular}{rrrrrrrrr}
    \toprule
    dataset    &  \multicolumn{2}{c}{Data size} & \multicolumn{2}{c}{NZDD size} &  \multicolumn{2}{c}{Variables} & \multicolumn{2}{c}{Constraints}\\
               &      $n$ &          $m$ &    $|V|$ &    $|E|$ &      Original &    Extended  & Original     &  Extended  \\
               \midrule 
    a9a        &    $123$ &     $32,561$ &    $775$ &  $20,657$ &     $32,685$ &     $21,556$ &     $65,123$ &  $41,317$  \\
    art-100000 &     $20$ &    $100,000$ &  $4,202$ &  $55,163$ &    $100,021$ &     $59,386$ &    $200,001$ & $110,329$  \\
    real-sim   & $20,955$ &     $72,309$ &     $38$ &   $7,922$ &     $93,265$ &     $28,916$ &    $144,619$ &  $15,847$  \\
    w8a        &    $300$ &     $49,749$ &    $209$ &  $34,066$ &     $50,050$ &     $34,576$ &     $99,499$ &  $68,135$  \\ 
    HIGGS      &     $28$ & $11,000,000$ &    $151$ &     $989$ & $11,000,029$ &      $1,169$ & $22,000,001$ &   $1,981$  \\
    \bottomrule
\end{tabular}